\documentclass[10pt, journal]{IEEEtran}
\usepackage{amsmath,amssymb,amsfonts}
\usepackage{algorithmic}
\usepackage{array}
\usepackage[caption=false,labelfont={},textfont={}]{subfig}
\usepackage{textcomp}
\usepackage{stfloats}
\usepackage{url}
\usepackage{verbatim}
\usepackage{graphicx}
\usepackage{cite}
\usepackage{xcolor}
\usepackage{makecell}
\usepackage{subfig}

\def\BibTeX{{\rm B\kern-.05em{\sc i\kern-.025em b}\kern-.08em
    T\kern-.1667em\lower.7ex\hbox{E}\kern-.125emX}}
\usepackage{balance}

\bstctlcite{IEEEexample:BSTcontrol}

\newtheorem{theorem}{Theorem}

\newtheorem{proposition}{Proposition}

\newtheorem{remark}{Remark}
\newtheorem{definition}{Definition}
\DeclareMathAlphabet{\mathbbold}{U}{bbold}{m}{n}

\usepackage{algorithmic}
\usepackage{algorithm}

\floatstyle{ruled}
\newfloat{algorithm}{tbp}{loa}
\providecommand{\algorithmname}{Algorithm}
\floatname{algorithm}{\protect\algorithmname}

\begin{document}
\title{Privacy-Preserving Localization via Transmit Antenna Selection and Permutation
}

\author{Yiyang Zhang, Yanmo Hu, Junyuan Gao, Shuowen Zhang, \textit{Senior Member, IEEE},\\ Jiannong Cao, \textit{Fellow, IEEE}, and Liang Liu, \textit{Fellow, IEEE}
\thanks{Yiyang Zhang, Yanmo Hu, Junyuan Gao, Shuowen Zhang, Jiannong Cao, and Liang Liu are with The Hong Kong Polytechnic University, Hong Kong SAR (e-mail: yiyang-eee.zhang@connect.polyu.hk,\! \{yanmo.hu,\! junyuan.gao,\! shuowen.zhang,\! jiannong.cao,\! liang-eie.liu\}@polyu.edu.hk).}
\thanks{An earlier version of this paper was presented in part at the 2025 IEEE Globecom Workshops \cite{Zhang_2025_TransmitAntennaSelectionforPreserving6GLocalizationPrivacy}.}
}

\maketitle

\begin{abstract}
Integrated sensing and communication (ISAC) has been identified as one primary usage scenario in the sixth-generation (6G) network. While techniques to preserve information privacy, such as cryptography, have been widely investigated, how to preserve sensing privacy is still an open problem in the literature. This paper makes an early attempt to tackle the above issue. Specifically, we consider a localization system consisting of a multi-antenna transmitter, termed Alice, a single-antenna legitimate receiver, termed Bob, and a single-antenna illegitimate receiver, termed Eve. To allow Bob to estimate Alice's angle-of-departure (AOD) but prevent Eve from performing this task based on Alice's signals, this paper proposes a novel antenna selection and permutation based transmission strategy for Alice. Under this scheme, Alice carefully selects a subset of antennas and permutes their indices to establish a specific pilot-antenna mapping for transmission. Similar to cryptography for information privacy, such a mapping will serve as the secret key to preserve localization privacy. 
In the special case without noise at Bob and Eve, we manage to find out all the antenna selection and permutation solutions such that with this key (knowledge about the exact pilot-antenna mapping), Bob can uniquely estimate Alice's AOD, while without this key, Eve can estimate multiple AODs of Alice that can lead to its received signals. 
In the noisy case, numerical results are provided to show that our scheme can confuse Eve to make inaccurate AOD estimation as well.
\end{abstract}


\section{Introduction}
\subsection{Background and Motivation}
\IEEEPARstart{R}{ecently}, ITU-R has identified integrated sensing and communication (ISAC) as one of the six primary usage scenarios of the sixth-generation (6G) cellular network \cite{ITUR_2023_TheITURFrameworkforIMT2030}. This motivates tremendous efforts from both academia and industry to embed the sensing function into the communication-oriented cellular network \cite{Liu_2024_LeveragingaVarietyofAnchorsinCellularNetworkforUbiquitousSensing, Liu_2020_ATwoStageRadarSensingApproachBasedonMIMOOFDMTechnology,
Hu_2025_AnchorPointsAssistedUplinkSensinginPerceptiveMobileNetworks, 
Zhang_2022_EnablingJointCommunicationandRadarSensinginMobileNetworks—ASurvey,Rahal_2024_RISEnabledNLoSnearFieldJointPositionandVelocityEstimationunderUserMobility,Gao_2026_IntegratedMassiveCommunicationandTargetLocalizationin6GCellFreeNetworks, hu2026sarisarimaging6gnetwork,Gao_2026_MultiViewImaginginNetworkedSensingSystemsaCovarianceBasedApproach}. 
Such a trend, however, gives rise to new sensing privacy issues. For example, when a transmitter emits wireless signals, illegitimate receivers can monitor vital signs \cite{Liu_2016_ContactlessRespirationMonitoringViaofftheShelfWiFiDevices} and recognize daily activities \cite{Lu_2024_AnImperceptibleEavesdroppingAttackonWiFiSensingSystems} of people from the signals reflected from them. Therefore, in the future ISAC network, we need to prevent illegitimate receivers from not only eavesdropping on the messages but also sensing the environment based on their received signals. 
\begin{figure}[t]
    \centerline{\includegraphics[width = 0.97\linewidth]{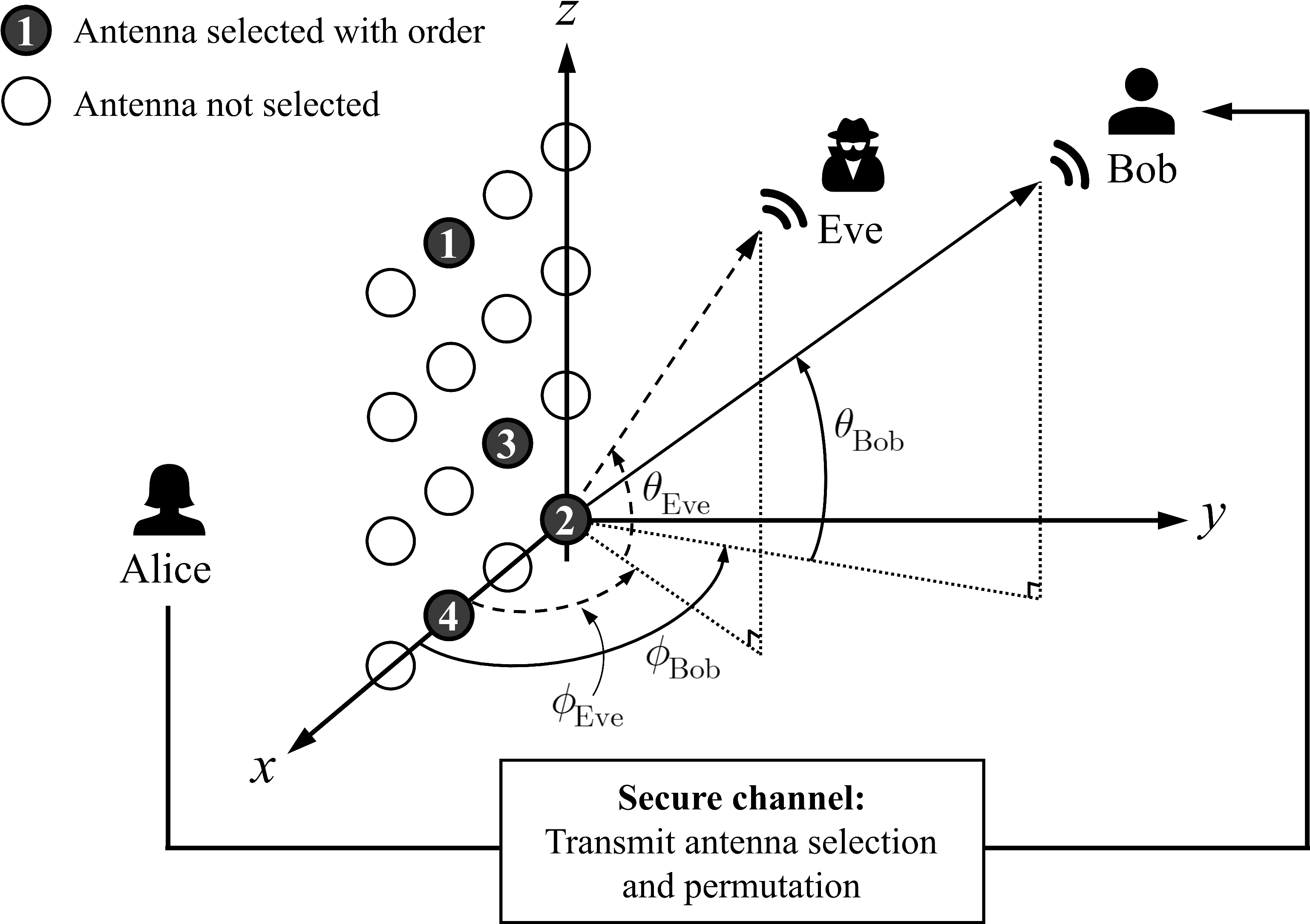}}
    \caption{Our considered privacy-preserving localization system. Alice properly selects a subset of antennas and permutes their indices to form a pilot-antenna mapping. In this example, there are four pilot sequences, and four antennas (black ones) are selected, while the antenna in the first row and the second column will transmit the first pilot sequence and so on. This antenna selection and permutation strategy, which serves as the secret key for localization, is then shared with Bob over a secure channel based on the cryptography technique. At last, Alice transmits pilot signals over these selected and permuted antennas. Our goal is to allow Bob to accurately estimate Alice's AOD with the secret key and prevent Eve from estimating Alice's AOD without the key.} 
    %\vspace{-10pt}
    \label{fig:system-model}
\end{figure}

While information privacy preserving techniques, such as cryptography, have been widely investigated and applied \cite{Yuan_2024_LowCostFederatedBroadLearningforPrivacyPreservedKnowledgeSharingintheRISAidedInternetofVehicles, Li_2025_IntelligentCovertCommunicationRecentAdvancesandFutureResearchTrends}, the research about 6G sensing privacy preservation is still in its infancy. This motivates us to study privacy-preserving sensing techniques in this paper. Specifically, this paper focuses on the localization application. In our considered localization system, as shown in Fig. \ref{fig:system-model}, there exist a multi-antenna transmitter, named Alice, a single-antenna legitimate receiver, named Bob, and a single-antenna illegitimate receiver, named Eve. After Alice emits a transmit signal, both Eve and Bob can apply the existing algorithms, such as maximum likelihood estimation \cite{Fascista_2019_MillimeterWaveDownlinkPositioningwithaSingleAntennaReceiver} and subspace-based methods \cite{Liao_2016_MUSICforSingleSnapshotSpectralEstimationStabilityandSuperResolution}, to estimate the angle-of-departure (AOD) information of Alice. We thus face the following problem: how to prevent Eve from localizing Alice based on Alice's signal?

Cryptography has been widely used to preserve information privacy. Note that cryptography is conducted at higher layers, not the physical layer. Specifically, an eavesdropper is able to recover the information bits based on physical-layer techniques such as demodulation and channel decoding. However, without the key, it is practically impossible for the eavesdropper to decode the messages from these bits. In sharp contrast to encryption that breaks the link between bits and messages, the privacy-preserving technique for sensing has to be done at the physical layer. This is because if the illegitimate receiver can have full knowledge about the channel state information (CSI), there is nothing that we can do at the higher layer to prevent it from recovering the location information from CSI. As a result, to preserve sensing privacy, we need to prevent Eve from knowing CSI at the physical layer.

Although higher-layer privacy-preserving techniques cannot be applied for localization, we do learn an important lesson from these techniques - we should use a secret key, which is some information that is known by Bob but unknown by Eve, to preserve localization privacy in our considered system. This gives rise to the following problem: what physical-layer information can serve as the key such that Bob can estimate Alice's AOD with it but Eve cannot without it?

\subsection{Related Works}
Recently, some pioneering works have been done to preserve privacy at the physical layer for various sensing applications \cite{ Wang_2024_MultiAntennaSignalMaskingandRoundTripTransmissionforPrivacyPreservingWirelessSensing, Luo_2024_MIMOCryptMultiUserPrivacyPreservingWiFiSensingViaMIMOEncryption, Hernandez_2023_ScheduledSpatialSensingagainstAdversarialWiFiSensing, Meng_2023_SecurFiaSecureWirelessSensingSystemBasedonCommercialWiFiDevices}.
In \cite{Wang_2024_MultiAntennaSignalMaskingandRoundTripTransmissionforPrivacyPreservingWirelessSensing}, secret time-varying phase shifts are applied to the pilot signals across transmit antennas to induce erroneous breathing and heart rate estimates at Eve, or to hide motion features against activity recognition by Eve. Similarly, \cite{Luo_2024_MIMOCryptMultiUserPrivacyPreservingWiFiSensingViaMIMOEncryption} changes both the amplitudes and phases of the pilot signals across different antennas over time to protect privacy-sensitive hand gesture recognition.
Beyond signal manipulation, both \cite{Hernandez_2023_ScheduledSpatialSensingagainstAdversarialWiFiSensing} and \cite{Meng_2023_SecurFiaSecureWirelessSensingSystemBasedonCommercialWiFiDevices} switch the transmit antenna based on a secret schedule to prevent unauthorized sensing in various applications, including human activity recognition and tracking.
Essentially, the above physical-layer approaches all preserve sensing privacy by obfuscating the CSI estimated at eavesdroppers.

With regard to privacy in localization, a fundamental approach is to add random noise and phase shifts, thereby making the distorted CSI difficult to be associated with a specific location \cite{Nguyen_2024_UnderstandingPrivacyRisksofHighAccuracyRadioPositioningandSensinginWirelessNetworks}. For instance, \cite{Ghiro_2022_OntheImplementationofLocationObfuscationinOpenwifiandItsPerformance} proposes to alter the amplitude and phase of the transmit signal based on a secret random process, which mitigates privacy risks in indoor fingerprint-based localization. Moreover, because localization can also be performed based on explicit spatial parameters extracted from the CSI, such as range and angle, a more targeted defense is to directly obfuscate these spatial parameters. Along this line, in \cite{Li_2024_ChannelStateInformationFreeLocationPrivacyEnhancementFakePathInjection}, Alice performs fake path injection via generating structured transmit signals, such that Eve finds it difficult to distinguish the true paths from the fake paths since the delay and angle information of these two
kinds of paths are close. However, even if Eve believes that Alice’s signals are from the fake paths, the estimated location of Alice is only subject to a marginal error. This approach is thus extended in \cite{Li_2024_ChannelStateInformationFreeLocationPrivacyEnhancementDelayAngleInformationSpoofing}, where the true paths are virtually shifted to fake ones via delay-angle spoofing, such that Eve perceives Alice at an incorrect location that could be far from the true one.
To preserve legitimate localization in \cite{Li_2024_ChannelStateInformationFreeLocationPrivacyEnhancementFakePathInjection, Li_2024_ChannelStateInformationFreeLocationPrivacyEnhancementDelayAngleInformationSpoofing}, a small secret message containing the fake path parameters is shared with Bob via a secure channel.

\subsection{Main Contributions}
In this paper, we consider a localization system consisting of Alice, Bob, and Eve, as shown in Fig. \ref{fig:system-model}. Our goal is to design Alice's transmission strategy together with the key such that Bob can estimate Alice's AOD accurately based on its received signal and the key, but Eve cannot without the key. The main contributions of this paper are summarized as follows. 

\begin{itemize}
    \item The first contribution is that this paper designs a transmit antenna selection and permutation strategy to achieve the above localization privacy preservation goal. Under this scheme, Alice first selects a subset of antennas for transmission, and then permutes the indices of these selected antennas to generate an ordered antenna sequence. The pre-designed pilot signals, which follow a fixed order, are sequentially paired to the selected antennas according to this ordered antenna sequence. This establishes an adaptive pilot-antenna mapping instead of a fixed one. Then, this antenna selection and permutation strategy will serve as the key and be only shared with Bob.
    Note that to estimate AOD information based on the existing algorithms \cite{Fascista_2019_MillimeterWaveDownlinkPositioningwithaSingleAntennaReceiver, Liao_2016_MUSICforSingleSnapshotSpectralEstimationStabilityandSuperResolution}, we need knowledge about pilot signals and steering vectors. 
    In \cite{Li_2024_ChannelStateInformationFreeLocationPrivacyEnhancementDelayAngleInformationSpoofing}, the pilot signals are changed such that they can serve as the key to prevent Eve from localizing Alice. However, standardized systems often require the use of predefined pilots \cite{Zhang_2025_PrivacyPreservationinMIMOOFDMLocalizationSystemsaBeamformingApproach}.
    Therefore, in this paper, we propose to hide the steering vector information from Eve. 

    \item Second, we present two criteria to evaluate the performance of sensing privacy preserving schemes - in the noiseless case, i.e., there is no noise at Bob and Eve, whether Bob can estimate Alice's AOD uniquely with the key, and whether Eve can have multiple AOD estimations that all result in its received signals due to the lack of the key.
    
    \item Third, we manage to characterize the set that consists of all the antenna selection and permutation strategies such that when an arbitrary strategy from this set is used by Alice for pilot transmission, Bob can always uniquely estimate Alice's AOD in the noiseless case, based on its received signals and knowledge about this strategy. The key to achieve this goal is to mitigate the grating lobe issue \cite{Mailloux_1994_PhasedArrayAntennaHandbook, Balanis_2005_ArraysLinearPlanarandCircular}. Improper antenna selection may result in excessively large inter-element spacing among the selected antennas, thereby introducing grating lobes such that distinct AODs lead to the same received signals. To avoid this, we establish a requirement on the index differences of the selected antennas that guarantees unique AOD estimation.
    
    \item Fourth, we manage to rigorously prove that if any antenna selection and permutation strategy in the above set is applied, Eve, who does not know which strategy is applied, can estimate multiple AODs that lead to its received signal in the noiseless case. Therefore, Eve does not know which AOD among the estimated ones is the true AOD. 
    Specifically, we reveal that the estimated AODs at Eve inherently occur in groups, each generally containing eight solutions. Given that any strategy in the above set causes Eve to estimate as few as one such group, we further construct a subset of antenna selection and permutation strategies to achieve enhanced privacy, such that when Alice applies an arbitrary strategy from this subset, Eve will always be confused by multiple groups of AODs, even when it knows this subset.

    \item Last, we provide numerical results to verify that in the noisy case, Bob can reliably recover Alice's AOD under our scheme, whereas Eve can easily mistake an incorrect AOD solution for the true one, resulting in a large estimation error even at high signal-to-noise ratios (SNRs). It is also shown that our proposed enhanced privacy design further reduces Eve's estimation accuracy.
\end{itemize}

Note that antenna selection has been widely studied for wireless communication \cite{Sanayei_2004_AntennaSelectioninMIMOSystems, Gore_2002_MIMOAntennaSubsetSelectionwithSpaceTimeCoding, Basar17}. Recently, the techniques of fluid antenna \cite{Wong21} and movable antenna \cite{Zhu24} can also be viewed as new ways to realize antenna selection, because essentially all of them need to determine the locations of the transmit antennas. However, to our best knowledge, our work is an early work to utilize antenna selection to preserve localization privacy in ISAC systems.

\textit{Organization: }The rest of this paper is organized as follows. Section \ref{sec:system-model} investigates the localization privacy issue and introduces transmit antenna selection and permutation strategy to preserve privacy. Section \ref{sec:privacy-design-statement} clarifies the goals of the privacy design. Sections \ref{sec:antenna-selection-strategy-design-uniqueness} and \ref{sec:antenna-selection-strategy-design-ambiguity} provide the guidelines to ensure uniqueness in Bob's AOD estimation, and to induce ambiguity in Eve's AOD estimation, respectively. Numerical results are shown in Section \ref{sec:numerical-results}, and Section \ref{sec:conclusions} concludes the paper.

\textit{Notations: }
Throughout this paper, $\operatorname{rank}(\cdot)$ denotes the rank, and $\operatorname{det}(\cdot)$ denotes the determinant of a square matrix.
The sum of two sets $\mathcal{X}$ and $\mathcal{Y}$ represents their Minkowski sum, i.e., $\mathcal{X} + \mathcal{Y} = \{\boldsymbol{x} + \boldsymbol{y} \mid \boldsymbol{x} \in \mathcal{X}, \boldsymbol{y} \in \mathcal{Y}\}$.
$\mathrm{GL}(n, \mathbb{Z}) = \{\mathbf{M} \in \mathbb{Z}^{n \times n} \mid \det(\mathbf{M}) = \pm 1\}$ denotes the general linear group of degree $n$ over $\mathbb{Z}$.

\section{System Model}\label{sec:system-model}

\subsection{Localization Privacy Issue} \label{subsec:localization-privacy-issue}
In this paper, we consider a privacy-preserving localization system, consisting of
a legitimate transmitter (Alice),
a legitimate receiver (Bob),
and an illegitimate device (Eve), as shown in Fig. \ref{fig:system-model}.
Alice is equipped with a uniform planar array (UPA) that is placed on the $xz$-plane and composed of $M=M_x \times M_z$ antennas, where $M_x$ and $M_z$ denote the numbers of antennas at the $x$-axis domain and the $z$-axis domain, respectively. We index the antenna at the $m_x$-th column and the $m_z$-th row of the array as antenna $(m_x,m_z)$, $m_x=1, \ldots, M_x$, $m_z=1, \ldots, M_z$.
Moreover, both Bob and Eve are equipped with one antenna.
The AOD from Alice to Bob is denoted by the elevation angle $\theta_{\mathrm{Bob}}\in (-\pi/2, \pi/2)$ and the azimuth angle $\phi_{\mathrm{Bob}}\in (0, \pi)$, while that from Alice to Eve is denoted by the elevation angle $\theta_{\mathrm{Eve}}\in (-\pi/2, \pi/2)$ and the azimuth angle $\phi_{\mathrm{Eve}}\in  (0, \pi)$.

Suppose antenna spacing at Alice is half of the signal wavelength. Then, the steering vector of Alice's UPA towards $(\theta, \phi)$ is given as $\boldsymbol{a}(\theta, \phi)= [a_{1,1}(\theta, \phi), \ldots, a_{1,M_z}(\theta, \phi),  \ldots,$
$a_{M_x,1}(\theta, \phi), \ldots, a_{M_x,M_z}(\theta, \phi)]^T \in \mathbb{C}^{M\times 1}$, where
\begin{align}
    & a_{m_x,m_z}(\theta, \phi) \nonumber \\
    & = e^{-j\pi\left((m_x-1) \cos \theta \cos \phi + (m_z-1) \sin \theta  \right)}, \ \forall m_x, m_z.
\end{align}
To focus on the localization privacy issue, this paper assumes that only line-of-sight (LOS) channels exist between Alice and Bob and between Alice and Eve. Then, the channels from Alice to Bob and from Alice to Eve are respectively given as
\begin{align}
    &\boldsymbol{h}_{\mathrm{Bob}}(\theta_{\mathrm{Bob}}, \phi_{\mathrm{Bob}})=  {\gamma_{\mathrm{Bob}}} \boldsymbol{a}(\theta_{\mathrm{Bob}}, \phi_{\mathrm{Bob}})
    \in \mathbb{C}^{M \times 1}, \\
    &\boldsymbol{h}_{\mathrm{Eve}}(\theta_{\mathrm{Eve}}, \phi_{\mathrm{Eve}}) ={\gamma_{\mathrm{Eve}}} \boldsymbol{a}(\theta_{\mathrm{Eve}}, \phi_{\mathrm{Eve}})
    \in \mathbb{C}^{M \times 1},
\end{align}
where $\gamma_{\mathrm{Bob}}$ and $\gamma_{\mathrm{Eve}}$ denote the unknown complex channel coefficients from Alice to Bob and from Alice to Eve.

In the above system, Alice transmits radio pilot signals to Bob, who can then localize Alice by estimating the AOD information about $\theta_{\mathrm{Bob}}$ and $\phi_{\mathrm{Bob}}$. 
Under the conventional localization scheme, Alice will send pilot signals across all its antennas to Bob. Define $\boldsymbol{S}=[\boldsymbol{s}_{1,1}, \ldots, \boldsymbol{s}_{1,M_z}, \ldots, \boldsymbol{s}_{M_x,1}, \ldots, \boldsymbol{s}_{M_x,M_z}]\in \mathbb{C}^{G \times M}$ as the overall pilot signal transmitted by all the antennas of Alice, where $G$ is the pilot sequence length, and $\boldsymbol{s}_{m_x,m_z}=[{s}_{m_x,m_z, 1}, \ldots, {s}_{m_x,m_z, G}]^T \in \mathbb{C}^{G \times 1}$ denotes the pilot signal transmitted by antenna $(m_x,m_z)$ of Alice, $\forall m_x, m_z$, with ${s}_{m_x,m_z, g}$ denoting the $g$-th pilot symbol, $g = 1, 2, \ldots, G$.
Then, the signals received by Bob and Eve are respectively expressed as
\begin{align}
    &\boldsymbol{r}_{\mathrm{Bob}} = \boldsymbol{S} \boldsymbol{h}_\mathrm{Bob}(\theta_{\mathrm{Bob}}, \phi_{\mathrm{Bob}}) + \boldsymbol{n}_{\mathrm{Bob}} \in \mathbb{C}^{G \times 1},  \label{equ:Bob-received-signal} \\
    &\boldsymbol{r}_{\mathrm{Eve}} = \boldsymbol{S} \boldsymbol{h}_\mathrm{Eve}(\theta_{\mathrm{Eve}}, \phi_{\mathrm{Eve}}) + \boldsymbol{n}_{\mathrm{Eve}} \in \mathbb{C}^{G \times 1},   \label{equ:Eve-received-signal}
\end{align}
where $\boldsymbol{n}_{\mathrm{Bob}}$ and $\boldsymbol{n}_{\mathrm{Eve}}$ denote the noise at Bob and Eve.
There thus exists the localization privacy issue: although Bob can apply classic methods, such as \cite{Fascista_2019_MillimeterWaveDownlinkPositioningwithaSingleAntennaReceiver, Liao_2016_MUSICforSingleSnapshotSpectralEstimationStabilityandSuperResolution,VanVeen_1988_BeamformingaVersatileApproachtoSpatialFiltering}, to estimate $\theta_{\mathrm{Bob}}$ and $\phi_{\mathrm{Bob}}$ from (\ref{equ:Bob-received-signal}), Eve can adopt the same approach to estimate $\theta_{\mathrm{Eve}}$ and $\phi_{\mathrm{Eve}}$ from (\ref{equ:Eve-received-signal}) for localizing Alice when $\boldsymbol{S}$ is also known by Eve. This gives rise to the following question: can we design a clever transmission strategy for Alice such that only Bob can localize Alice based on its received signals?

\subsection{Transmit Antenna Selection and Permutation Strategy}
The lesson from the above example is that Eve can localize Alice when it perfectly knows Alice's pilot signal. Therefore, we propose a transmit antenna selection and permutation strategy such that some randomness shared with Bob but unknown to Eve is added on the pilot signals from Alice. Specifically, Alice first randomly selects $K \leq M$ antennas out of all the $M$ antennas for transmission. Define $\tilde{\boldsymbol{S}} = [\tilde{\boldsymbol{s}}_1, \ldots, \tilde{\boldsymbol{s}}_K] \in \mathbb{C}^{G \times K}$ as the collection of the pre-designed $K$ pilot signals for the selected antennas. 
In this paper, we assume that $G \geq K$, which is a challenging setup such that in the absence of any privacy-preserving measures, Eve has sufficient observations to reliably recover the channel $\boldsymbol{h}_\mathrm{Eve}$ and thus estimate $\theta_{\mathrm{Eve}}$ and $\phi_{\mathrm{Eve}}$ with high accuracy.
Next, Alice randomly pairs each pilot in $\tilde{\boldsymbol{S}}$ to one selected antenna. In other words, Alice permutes these $K$ selected antennas to transmit $\tilde{\boldsymbol{S}}$. Define the index of the antenna to transmit $\tilde{\boldsymbol{s}}_k$ as $(x_k, z_k)$, $k = 1, \dots, K$, and
\begin{equation}
\boldsymbol{\Omega} = \big( (x_1, z_1), (x_2, z_2), \ldots, (x_K, z_K) \big), \label{equ:Omega-definition}
\end{equation}
as the ordered sequence of the $K$ selected antennas. Thus, the overall transmit pilot signal of all the $M$ antennas, i.e., $\boldsymbol{S}$, is given as\footnote{It is worth noting that when $K=M$, Alice uses all the antennas in the UPA for transmission as in the conventional scheme in Section \ref{subsec:localization-privacy-issue}. However, the mappings between pilot signals and antennas differ. In Section \ref{subsec:localization-privacy-issue}, these mappings are fixed, i.e., the pilot signal $\boldsymbol{s}_{m_x, m_z}$ is explicitly paired to antenna $(m_x, m_z)$. In contrast, our proposed strategy applies an adaptive permutation $\boldsymbol{\Omega}$ such that each pilot signal is no longer transmitted by a fixed antenna.}
\begin{equation}
    \boldsymbol{s}_{m_x,m_z} = \begin{cases}
    \tilde{\boldsymbol{s}}_k, &\mathrm{if} \ (m_x,m_z) = (x_k,z_k), \ 
     k = 1, \dots, K,\\
    \boldsymbol{0}, &\mathrm{otherwise}. 
    \end{cases}
\end{equation}

Given the above antenna selection and permutation strategy, the received signals of Bob and Eve given in (\ref{equ:Bob-received-signal}) and (\ref{equ:Eve-received-signal}) can then be expressed as
\begin{align}
&\tilde{\boldsymbol{r}}_{\mathrm{Bob}}(\boldsymbol{\Omega}, \theta_{\mathrm{Bob}}, \phi_{\mathrm{Bob}}) =\tilde{\boldsymbol{S}}\tilde{\boldsymbol{h}}_{\mathrm{Bob}}(\boldsymbol{\Omega}, \theta_{\mathrm{Bob}}, \phi_{\mathrm{Bob}}) + \boldsymbol{n}_{\mathrm{Bob}}, \label{equ:Bob-received-signal-with-selection}\\
&\tilde{\boldsymbol{r}}_{\mathrm{Eve}}(\boldsymbol{\Omega}, \theta_{\mathrm{Eve}}, \phi_{\mathrm{Eve}}) =\tilde{\boldsymbol{S}}\tilde{\boldsymbol{h}}_{\mathrm{Eve}}(\boldsymbol{\Omega}, \theta_{\mathrm{Eve}}, \phi_{\mathrm{Eve}}) + \boldsymbol{n}_{\mathrm{Eve}}, \label{equ:Eve-received-signal-with-selection}
\end{align}
where
\begin{align}  
    &\tilde{\boldsymbol{h}}_{\mathrm{Bob}}(\boldsymbol{\Omega}, \theta_{\mathrm{Bob}}, \phi_{\mathrm{Bob}}) = {\gamma_{\mathrm{Bob}}} \big[  a_{x_1, z_1}(\theta_{\mathrm{Bob}}, \phi_{\mathrm{Bob}}), \ldots,  \nonumber \\
    & \hspace{90pt} a_{x_K, z_K}(\theta_{\mathrm{Bob}}, \phi_{\mathrm{Bob}})\big]^T \in \mathbb{C}^{K \times 1} ,   \label{equ:Bob-tilde-h-B-theta-phi}
    \\
    &\tilde{\boldsymbol{h}}_{\mathrm{Eve}}(\boldsymbol{\Omega}, \theta_{\mathrm{Eve}}, \phi_{\mathrm{Eve}}) = {\gamma_{\mathrm{Eve}}} [ a_{x_1, z_1}(\theta_{\mathrm{Eve}}, \phi_{\mathrm{Eve}}), \ldots, \nonumber \\
    &\hspace{88pt}  a_{x_K, z_K}(\theta_{\mathrm{Eve}}, \phi_{\mathrm{Eve}})]^T \in \mathbb{C}^{K \times 1}, \label{equ:Eve-tilde-h-B-theta-phi}
\end{align}
denote the channel vectors from the selected and permuted antennas of Alice to Bob and Eve, respectively.

Under our proposed privacy-preserving scheme, the antenna selection and permutation strategy $\boldsymbol{\Omega}$ as given in (\ref{equ:Omega-definition}), which specifies both the selected antenna subset and their permutation order used for pilot transmission, is securely sent to Bob based on the cryptography technique, while it is unknown to Eve. In other words, which antenna transmits $\tilde{\boldsymbol{s}}_k$, $\forall k$, i.e., $x_1, z_1,\dots,x_K,z_K$ in (\ref{equ:Bob-tilde-h-B-theta-phi}) and (\ref{equ:Eve-tilde-h-B-theta-phi}), is known by Bob but unknown by Eve. The other parameters, e.g., the UPA size $M_x \times M_z$, the value of selected antennas $K$, and the pre-designed pilot signal $\tilde{\boldsymbol{S}}$, can be known by both Bob and Eve. 
In the rest of this paper, we show how to design the antenna selection and permutation strategy $\boldsymbol{\Omega}$ to achieve privacy-preserving localization.

\section{Privacy Design Statement}  \label{sec:privacy-design-statement}

Under cryptography for secure communication, the core is to design the key such that the legitimate receiver with the key can successfully decode the message, while the eavesdropper without the key cannot. In our considered privacy-preserving localization scheme, the transmit antenna selection and permutation strategy $\boldsymbol{\Omega}$ serves as the role of key. Our goals are two-fold: when $\boldsymbol{\Omega}$ is known, it is easy for Bob to uniquely estimate $\theta_{\mathrm{Bob}}$ and $\phi_{\mathrm{Bob}}$ from (\ref{equ:Bob-received-signal-with-selection}); while when $\boldsymbol{\Omega}$ is unknown, it is infeasible for Eve to uniquely estimate $\theta_{\mathrm{Eve}}$ and $\phi_{\mathrm{Eve}}$ from (\ref{equ:Eve-received-signal-with-selection}). In the following, we define these two goals more precisely.  
\begin{definition}[Uniqueness in Bob's AOD Estimation] \label{def:Bob-uniqueness}
    An antenna selection and permutation strategy $\boldsymbol{\Omega}$ is defined to satisfy Bob's estimation uniqueness requirement if and only if given any ${\theta}_{\mathrm{Bob}}, \tilde{\theta}_{\mathrm{Bob}} \in (-\pi/2, \pi/2)$, and ${\phi}_{\mathrm{Bob}}, \tilde{\phi}_{\mathrm{Bob}}\in (0, \pi)$, the following condition holds:
    % for any $\tilde{\theta}_{\mathrm{Bob}} \in (-\pi/2, \pi/2)$, $\tilde{\phi}_{\mathrm{Bob}} \in (0, \pi)$ with $(\tilde{\theta}_{\mathrm{Bob}}, \tilde{\phi}_{\mathrm{Bob}}) \neq ({\theta}_{\mathrm{Bob}}, {\phi}_{\mathrm{Bob}})$
    \begin{align}
        \tilde{\boldsymbol{S}}\tilde{\boldsymbol{h}}(\boldsymbol{\Omega}, \tilde{\theta}_{\mathrm{Bob}}, \tilde{\phi}_{\mathrm{Bob}}) & \neq \tilde{\boldsymbol{S}} \tilde{\boldsymbol{h}}(\boldsymbol{\Omega}, {\theta}_{\mathrm{Bob}}, {\phi}_{\mathrm{Bob}}) \nonumber \\
        &\ \mathrm{if}\  (\tilde{\theta}_{\mathrm{Bob}}, \tilde{\phi}_{\mathrm{Bob}}) \neq ({\theta}_{\mathrm{Bob}}, {\phi}_{\mathrm{Bob}}).   \label{equ:requirement-def1}
    \end{align}
    In other words, in the ideal case without noise, Bob can uniquely estimate Alice's AOD information based on its received signal (\ref{equ:Bob-received-signal-with-selection}). 
    Moreover, define
    \begin{equation}
    \boldsymbol{\Lambda}=\{\boldsymbol{\Omega} \mid  \boldsymbol{\Omega} \ \mathrm{satisfies} \ (\ref{equ:requirement-def1})\}, \label{equ:Lambda-definition}
    \end{equation}
    as the set consisting of all the antenna selection and permutation strategies to ensure unique AOD estimation of Bob.
\end{definition}

Not surprisingly, not all the transmit antenna selection and permutation strategies can satisfy requirement (\ref{equ:requirement-def1}), as will be shown later in Section \ref{sec:antenna-selection-strategy-design-uniqueness}. This is because of the grating lobe effect \cite{Balanis_2005_ArraysLinearPlanarandCircular, Mailloux_1994_PhasedArrayAntennaHandbook}. Specifically, selecting a subset of antennas from the UPA increases the effective spacing between the antenna elements. It is well established that if this inter-element spacing becomes excessive, it induces spatial aliasing, resulting in grating lobes and consequently leading to directional ambiguities\cite{ Bray_2002_OptimizationofThinnedAperiodicLinearPhasedArraysUsingGeneticAlgorithmstoReduceGratingLobesduringScanning}. In practice, we have to select $\boldsymbol{\Omega}$ from $\boldsymbol{\Lambda}$ such that Bob can uniquely localize Alice. However, Eve also has the side information that $\boldsymbol{\Omega}$ must be chosen from $\boldsymbol{\Lambda}$, although the exact $\boldsymbol{\Omega}$ is unknown.
In the following, we define the requirement about $\boldsymbol{\Omega}$ to prevent Eve from estimating Alice's AOD information even with the above side information. 

\begin{definition}[Ambiguity in Eve’s AOD Estimation]\label{def:Eve-ambiguity}
    An antenna selection and permutation strategy $\boldsymbol{\Omega} \in \boldsymbol{\Lambda}$ (we will characterize $\boldsymbol{\Lambda}$ later in Theorems \ref{thm:equ-thm-uniqueness-K-equal-3} and \ref{thm:K-greater-than-3}) is defined to satisfy Eve's estimation ambiguity requirement if and only if given any ${\theta}_{\mathrm{Eve}} \in (-\pi/2, \pi/2)$ and ${\phi}_{\mathrm{Eve}} \in (0, \pi)$, there exists at least an alternative solution of $\tilde{\boldsymbol{\Omega}} \in \boldsymbol{\Lambda}$, $\tilde{\theta}_{\mathrm{Eve}} \in (-\pi/2, \pi/2)$ and $\tilde{\phi}_{\mathrm{Eve}}\in (0, \pi)$ with $(\tilde{\theta}_{\mathrm{Eve}}, \tilde{\phi}_{\mathrm{Eve}})  \neq ({\theta}_{\mathrm{Eve}}, {\phi}_{\mathrm{Eve}})$ such that
    \begin{equation}
        \tilde{\boldsymbol{S}}\tilde{\boldsymbol{h}}_{\mathrm{Eve}}(\tilde{\boldsymbol{\Omega}}, \tilde{\theta}_{\mathrm{Eve}}, \tilde{\phi}_{\mathrm{Eve}}) = \tilde{\boldsymbol{S}}\tilde{\boldsymbol{h}}_{\mathrm{Eve}}(\boldsymbol{\Omega}, {\theta}_{\mathrm{Eve}}, {\phi}_{\mathrm{Eve}}).
        \label{equ:requirement-def2}
    \end{equation}
    In other words, even if Eve knows that Alice must choose $\boldsymbol{\Omega}$ from the set $\boldsymbol{\Lambda}$ to satisfy Bob's estimation uniqueness requirement, there are multiple solutions of Alice's elevation angle and azimuth angle to make Eve get the same received signal in the ideal case without noise. 
\end{definition}

Note that under the cryptography technique for secure communication, without the knowledge about the key, the eavesdropper is theoretically able to decode the messages, but of prohibitively high complexity.
Under our considered privacy-preserving localization scheme, as shown in Definitions \ref{def:Bob-uniqueness} and \ref{def:Eve-ambiguity}, strict privacy can be preserved because theoretically speaking, distinct AODs of Alice can lead to the same received signal at Eve. 

\begin{remark}
To define uniqueness for Bob's estimation and ambiguity for Eve's estimation, we assume that there is no noise at Bob and Eve. After some $\boldsymbol{\Omega}$ is designed to satisfy the requirements of Definitions \ref{def:Bob-uniqueness} and \ref{def:Eve-ambiguity}, it can be directly applied in practice with noise at Bob and Eve. As will be shown in Section \ref{sec:numerical-results}, in the noisy case, our designed antenna selection and permutation strategy can still enable Bob to accurately estimate ${\theta}_{\mathrm{Bob}}$ and ${\phi}_{\mathrm{Bob}}$, and Eve to encounter significant ambiguity in estimating ${\theta}_{\mathrm{Eve}}$ and ${\phi}_{\mathrm{Eve}}$.
\end{remark}

For convenience, we introduce some notations that will be useful to show how to design the antenna selection and permutation strategy to satisfy the requirements in Definitions \ref{def:Bob-uniqueness} and \ref{def:Eve-ambiguity}. First, define
\begin{align}
    &u_{\mathrm{Bob}} = \cos (\theta_{\mathrm{Bob}}) \cos (\phi_{\mathrm{Bob}}), \  v_{\mathrm{Bob}} = \sin (\theta_{\mathrm{Bob}}), \label{equ:u-cos-cos-v-sin} \\ 
    &u_{\mathrm{Eve}} = \cos (\theta_{\mathrm{Eve}}) \cos (\phi_{\mathrm{Eve}}), \ v_{\mathrm{Eve}} = \sin (\theta_{\mathrm{Eve}}). \label{equ:u-cos-cos-v-sin-Eve}
\end{align}
Note that $(u_{\mathrm{Bob}}, v_{\mathrm{Bob}})$ and $(\theta_{\mathrm{Bob}}, \phi_{\mathrm{Bob}})$ are interchangeable when ${\theta}_{\mathrm{Bob}} \in (-\pi/2, \pi/2)$, ${\phi}_{\mathrm{Bob}} \in (0, \pi)$, and $0<u_{\mathrm{Bob}}^2 + v_{\mathrm{Bob}}^2 < 1$; $(u_{\mathrm{Eve}}, v_{\mathrm{Eve}})$ and $(\theta_{\mathrm{Eve}}, \phi_{\mathrm{Eve}})$ are interchangeable when ${\theta}_{\mathrm{Eve}} \in (-\pi/2, \pi/2)$, ${\phi}_{\mathrm{Eve}} \in (0, \pi)$, and $0< u_{\mathrm{Eve}}^2 + v_{\mathrm{Eve}}^2 < 1$. This is because given any $\theta_{\mathrm{Bob}} \in (-\pi/2, \pi/2)$ and $\phi_{\mathrm{Bob}} \in (0, \pi)$, $(u_{\mathrm{Bob}}, v_{\mathrm{Bob}})$ can be uniquely mapped based on (\ref{equ:u-cos-cos-v-sin}) that satisfies $0<u_{\mathrm{Bob}}^2 + v_{\mathrm{Bob}}^2 < 1$. Meanwhile, given any $(u_{\mathrm{Bob}}, v_{\mathrm{Bob}})$ with $0<u_{\mathrm{Bob}}^2+v_{\mathrm{Bob}}^2 < 1$, we can also uniquely obtain $\theta_{\mathrm{Bob}} = \arcsin(v_{\mathrm{Bob}}) \in (-\pi/2, \pi/2)$, and $\phi_{\mathrm{Bob}} = \arccos(u_{\mathrm{Bob}} / \sqrt{1 - v_{\mathrm{Bob}}^2}) \in (0, \pi)$ according to (\ref{equ:u-cos-cos-v-sin}). The case for Eve is analogous.
Due to the above one-to-one mapping, we can also claim that Definitions \ref{def:Bob-uniqueness} and \ref{def:Eve-ambiguity} are to make Bob uniquely estimate $u_{\mathrm{Bob}}$ and $v_{\mathrm{Bob}}$ with $0<u_{\mathrm{Bob}}^2+v_{\mathrm{Bob}}^2 < 1$, and Eve have ambiguity in estimating $u_{\mathrm{Eve}}$ and $v_{\mathrm{Eve}}$ with $0 < u_{\mathrm{Eve}}^2+v_{\mathrm{Eve}}^2 < 1$, in the noiseless case.

Second, given any antenna selection and permutation strategy $\boldsymbol{\Omega}$ as defined in (\ref{equ:Omega-definition}), define 
\begin{align}
&\boldsymbol{x}(\boldsymbol{\Omega}) = [0, x_2 - x_1, \ldots, x_K - x_1]^T, \label{equ:definition-xB}\\
&\boldsymbol{z}(\boldsymbol{\Omega}) = [0, z_2 - z_1, \ldots, z_K - z_1]^T, \label{equ:definition-zB}
\end{align}
as the vectors consisting of the index differences between any antenna and the first antenna in $\boldsymbol{\Omega}$ along the \(x\)-axis and the \(z\)-axis, respectively. Because AOD estimation is based on the phase difference among the transmit signals of Alice's antennas, as will be shown later, all the requirements on $\boldsymbol{\Omega}$ in Definitions \ref{def:Bob-uniqueness} and \ref{def:Eve-ambiguity} are on $\boldsymbol{x}(\boldsymbol{\Omega})$ and $\boldsymbol{z}(\boldsymbol{\Omega})$. 

Given the above notations, $\tilde{\boldsymbol{h}}_{\mathrm{Bob}}(\boldsymbol{\Omega}, \theta_{\mathrm{Bob}}, \phi_{\mathrm{Bob}})$ given in (\ref{equ:Bob-tilde-h-B-theta-phi}) and $\tilde{\boldsymbol{h}}_{\mathrm{Eve}}(\boldsymbol{\Omega}, \theta_{\mathrm{Eve}}, \phi_{\mathrm{Eve}})$ given in (\ref{equ:Eve-tilde-h-B-theta-phi}) can be expressed as
\begin{align}  
    & \tilde{\boldsymbol{h}}_{\mathrm{Bob}}(\boldsymbol{\Omega}, u_{\mathrm{Bob}}, v_{\mathrm{Bob}})= {\gamma_{\mathrm{Bob}}} e^{-j\pi \left((x_1-1) u_{\mathrm{Bob}} + (z_1-1) v_{\mathrm{Bob}}\right)} \nonumber \\  
    & \times \Big[1, \ldots,
    e^{-j \pi\left((x_K - x_1) u_{\mathrm{Bob}} + (z_K - z_1) v_{\mathrm{Bob}}\right)}\Big]^T,   \label{equ:Bob-tilde-h-B-u-v}   \\
    & \tilde{\boldsymbol{h}}_{\mathrm{Eve}}(\boldsymbol{\Omega}, u_{\mathrm{Eve}}, v_{\mathrm{Eve}})= {\gamma_{\mathrm{Eve}}}  e^{-j\pi \left((x_1-1) u_{\mathrm{Eve}} + (z_1-1) v_{\mathrm{Eve}}\right)} \nonumber \\
    & \times \Big[1, \ldots, e^{-j \pi\left((x_K - x_1) u_{\mathrm{Eve}} + (z_K - z_1) v_{\mathrm{Eve}}\right)}\Big]^T. &
    \label{equ:Eve-tilde-h-B-u-v}
\end{align}
Note that we replace $\theta_{\mathrm{Bob}}$, $\phi_{\mathrm{Bob}}$ by  $u_{\mathrm{Bob}}$, $v_{\mathrm{Bob}}$, and $\theta_{\mathrm{Eve}}$, $\phi_{\mathrm{Eve}}$ by $u_{\mathrm{Eve}}$, $v_{\mathrm{Eve}}$ in the above, because of their one-to-one mapping. 

Then, in the following two sections, we respectively introduce how to characterize and obtain $\boldsymbol{\Omega}$ based on Definition~\ref{def:Bob-uniqueness} to guarantee uniqueness in Bob’s AOD estimation, and based on Definition \ref{def:Eve-ambiguity} to simultaneously guarantee ambiguity in Eve's AOD estimation.

\section{Antenna Selection And Permutation Strategy Design For Bob's Estimation Uniqueness} \label{sec:antenna-selection-strategy-design-uniqueness}

Given any antenna selection and permutation strategy $\boldsymbol{\Omega}$, verifying whether it satisfies Bob’s estimation uniqueness requirement as in Definition \ref{def:Bob-uniqueness} is highly challenging. Specifically, since the parameters ${\theta}_{\mathrm{Bob}}$, $\tilde{\theta}_{\mathrm{Bob}}$, ${\phi}_{\mathrm{Bob}}$ and $\tilde{\phi}_{\mathrm{Bob}}$ are all defined over continuous domains, it is computationally intractable to verify whether condition (\ref{equ:requirement-def1}) holds for all possible pairs of distinct $(\tilde{\theta}_{\mathrm{Bob}}, \tilde{\phi}_{\mathrm{Bob}})$ and $(\theta_{\mathrm{Bob}}, \phi_{\mathrm{Bob}})$.
Therefore, in this section, we aim to provide an implementable approach to find all the antenna selection and permutation strategies that satisfy  Definition \ref{def:Bob-uniqueness}.

Given the notations established in the preceding section and by exploiting the periodicity of the complex exponential functions, the set $\boldsymbol{\Lambda}$ in (\ref{equ:Lambda-definition}) is characterized in the following proposition.

\begin{proposition} \label{prop:uniqueness}
    In the regime of $G \geq K$, i.e., the pilot sequence length is no smaller than the number of selected antennas, the set consisting of all the antenna selection and permutation strategies that satisfy (\ref{equ:requirement-def1}) in Definition \ref{def:Bob-uniqueness}, i.e., $\boldsymbol{\Lambda}$ given in (\ref{equ:Lambda-definition}), is characterized as
    \begin{equation}
    \begin{aligned}
    \boldsymbol{\Lambda}=\{\boldsymbol{\Omega} \mid \ &\mathrm{given\ any}\  u \ \mathrm{and}\  v \ \mathrm{with} \ 0 < u^2 + v^2 < 1, \\
    &\mathrm{at\ least\ one\ element\ in\ the\ vector}\\
    &u \boldsymbol{x}(\boldsymbol{\Omega})+ v \boldsymbol{z}(\boldsymbol{\Omega})\ \mathrm{is\ not\ an\ integer}\}.
        \end{aligned} \label{equ:theorem1}
    \end{equation}
\end{proposition}
\begin{IEEEproof}
See Appendix \ref{app:proof-theorem-Bob}.
\end{IEEEproof}

The antenna selection and permutation strategy characterized in Proposition \ref{prop:uniqueness} can remove the grating lobe issue \cite{Balanis_2005_ArraysLinearPlanarandCircular} for Bob. 
Actually, if we can find $u$ and $v$ such that all the elements in $u \boldsymbol{x}(\boldsymbol{\Omega})+ v \boldsymbol{z}(\boldsymbol{\Omega})$ are integers, given $u_{\mathrm{Bob}}$ and $v_{\mathrm{Bob}}$, the new solution of $u_{\mathrm{Bob}} + 2u$ and $v_{\mathrm{Bob}} + 2v$ will add a phase difference of $2\pi$ to all the elements in the channel $\tilde{\boldsymbol{h}}_{\mathrm{Bob}}(\boldsymbol{\Omega}, u_{\mathrm{Bob}}, v_{\mathrm{Bob}})$ given in (\ref{equ:Bob-tilde-h-B-u-v}). This may lead to the same received signal because of the periodicity of the complex exponential function. 
One example is as follows. Suppose the ordered sequence of selected antennas $\boldsymbol{\Omega}=((1, 1), (1, 4), (2, 3))$ such that $\boldsymbol{x}(\boldsymbol{\Omega}) = [0, 0, 1]^T$ and $\boldsymbol{z}(\boldsymbol{\Omega}) = [0, 3, 2]^T$. In this case, there exist $u=v=1/3$ such that $u \boldsymbol{x}(\boldsymbol{\Omega})+ v \boldsymbol{z}(\boldsymbol{\Omega}) = [0,1,1]^T$. Suppose $\theta_{\mathrm{Bob}} = \arcsin(-1/3)$ and $\phi_{\mathrm{Bob}}=\arccos(-\sqrt{2}/4)$ such that $u_{\mathrm{Bob}}=-1/3$ and $v_{\mathrm{Bob}}= -1/3$. We can construct a new solution of $\tilde{u}_{\mathrm{Bob}}=u_{\mathrm{Bob}}+2u= 1/3$ and $\tilde{v}_{\mathrm{Bob}}=v_{\mathrm{Bob}}+2v=1/3$. Correspondingly, $\tilde{\theta}_{\mathrm{Bob}}= \arcsin(1/3) = -\theta_{\mathrm{Bob}}$ and $\tilde{\phi}_{\mathrm{Bob}}=\arccos(\sqrt{2}/4) = \pi - \phi_{\mathrm{Bob}}$. It can be shown that $\tilde{\theta}_{\mathrm{Bob}}$, $\tilde{\phi}_{\mathrm{Bob}}$, $\theta_{\mathrm{Bob}}$ and $\phi_{\mathrm{Bob}}$ do not satisfy (\ref{equ:requirement-def1}) in Definition~\ref{def:Bob-uniqueness}. Therefore, Bob cannot uniquely estimate Alice's AOD.

For convenience, define
\begin{align}
    &\boldsymbol{w} = [u, v]^T, \\ 
    &\boldsymbol{P}(\boldsymbol{\Omega}) = [\boldsymbol{x}(\boldsymbol{\Omega}), \boldsymbol{z}(\boldsymbol{\Omega})] \in \mathbb{Z}^{K \times 2}, \label{equ:P-Omega}
\end{align}
such that $u \boldsymbol{x}(\boldsymbol{\Omega})+ v \boldsymbol{z}(\boldsymbol{\Omega}) = \boldsymbol{P}(\boldsymbol{\Omega})\boldsymbol{w}$. 
In the following, we will tackle the above challenge in obtaining the set $\boldsymbol{\Lambda}$ defined in (\ref{equ:theorem1}) under the cases when $K=3$ and $K>3$, respectively. 

\subsection{Case When $K=3$} \label{subsec:case-when-K-3}
Under the special case when $K=3$, we are able to characterize $\boldsymbol{\Lambda}$ based on the following theorem.

\begin{theorem}\label{thm:equ-thm-uniqueness-K-equal-3}
In the case when $K=3$ antennas are selected from $M \geq 3$ antennas, the set $\boldsymbol{\Lambda}$ in (\ref{equ:theorem1}) is characterized as
\begin{equation}  \label{equ:Lambda-K-equal-3}
    \boldsymbol{\Lambda} = \{\boldsymbol{\Omega} \mid |\det([\boldsymbol{P}(\boldsymbol{\Omega})]_{2:3, :})| = 1\}, 
\end{equation}
where $\boldsymbol{P}(\boldsymbol{\Omega})$ is defined in (\ref{equ:P-Omega}), and $[\boldsymbol{P}(\boldsymbol{\Omega})]_{2:3, :}$ denotes the submatrix formed by the second and third rows of $\boldsymbol{P}(\boldsymbol{\Omega})$.
\end{theorem}

\begin{IEEEproof}
    See Appendix \ref{app:equ-thm-uniqueness-K-equal-3}.
\end{IEEEproof}

The condition $|\det([\boldsymbol{P}(\boldsymbol{\Omega})]_{2:3, :})| = 1$ in Theorem \ref{thm:equ-thm-uniqueness-K-equal-3} implies that $[\boldsymbol{P}(\boldsymbol{\Omega})]_{2:3, :}$ is unimodular. Geometrically, a unimodular transformation is lattice-preserving, meaning $[\boldsymbol{P}(\boldsymbol{\Omega})]_{2:3, :}$ bijectively maps the integer lattice $\mathbb{Z}^2$ onto itself. Since the region $\{\boldsymbol{w} \in \mathbb{R}^{2} \mid 0< \|\boldsymbol{w}\|_2<1\}$ contains no integer point, it follows that the region $\{[\boldsymbol{P}(\boldsymbol{\Omega})]_{2:3, :}\boldsymbol{w} \mid \boldsymbol{w} \in \mathbb{R}^{2}, 0 < \|\boldsymbol{w}\|_2<1\}$ contains no integer point.

\subsection{Case When $K>3$}

Under the case when $K>3$, in general, we do not have a closed-form solution to characterize $\boldsymbol{\Lambda}$. In the following, we propose an efficient algorithm to obtain all $\boldsymbol{\Omega}$ in $\boldsymbol{\Lambda}$. 
We formulate the following feasibility problem:
\begin{subequations} \label{equ:feasibility-problem-of-w}
\begin{alignat}{2}
    & \mathrm{find} \quad && \boldsymbol{w} \\
    & \mathrm{subject\ to} \quad && \boldsymbol{P}(\boldsymbol{\Omega}) \boldsymbol{w} \in \mathbb{Z}^{K \times 1}, \label{equ:constraint-Pw-in-ZK}\\
    & && 0 < \|\boldsymbol{w}\|^2 < 1.  \label{equ:constraint-w}
\end{alignat}
\end{subequations}
Given some $\boldsymbol{\Omega}$, if the above problem admits no solution, we can conclude that $\boldsymbol{\Omega} \in \boldsymbol{\Lambda}$.

First, any $\boldsymbol{\Omega}$ with $\operatorname{rank}(\boldsymbol{P}(\boldsymbol{\Omega})) = 1$ must satisfy $\boldsymbol{\Omega} \notin \boldsymbol{\Lambda}$, because we can always find some $\boldsymbol{w}$ with arbitrarily small norm power that makes $\boldsymbol{P}(\boldsymbol{\Omega})\boldsymbol{w} = \mathbf{0}$.

Second, given some $\boldsymbol{\Omega}$ with $\operatorname{rank}(\boldsymbol{P}(\boldsymbol{\Omega})) = 2$, we solve the above feasibility problem as follows. 
Specifically, according to Cauchy-Schwarz inequality and (\ref{equ:constraint-w}), we have
\begin{align}
    |[\boldsymbol{P}(\boldsymbol{\Omega})]_{k,:} \boldsymbol{w}| & \leq \|[\boldsymbol{P}(\boldsymbol{\Omega})]_{k,:}\| \|\boldsymbol{w}\|  \nonumber \\
    & < \sqrt{(M_x-1)^2 + (M_z-1)^2},
\end{align}
where $[\boldsymbol{P}(\boldsymbol{\Omega})]_{k,:}$ denotes the $k$-th row of $\boldsymbol{P}(\boldsymbol{\Omega})$.
Define
\begin{align} \label{equ:search-space-n-2D}
    \mathcal{N} = \big \{ &[n_1, \dots, n_K]^T \in \mathbb{Z}^{K \times 1} \mid \nonumber \\
    & |n_k| < \sqrt{(M_x-1)^2 + (M_z-1)^2}, \forall k \big\}.
\end{align}
Then, problem (\ref{equ:feasibility-problem-of-w}) is equivalent to the following problem
\begin{subequations} \label{equ:feasibility-problem-of-w-Pw-in-N}
\begin{alignat}{2}
    & \mathrm{find} \quad && \boldsymbol{w} \\
    & \mathrm{subject\ to} \quad && \boldsymbol{P}(\boldsymbol{\Omega})\boldsymbol{w} \in \mathcal{N}, \label{equ:Pw-in-N}\\
    & && (\text{\ref{equ:constraint-w}}) \label{equ:constraint-w-2}.
\end{alignat}
\end{subequations}
Different from (\ref{equ:constraint-Pw-in-ZK}) where there are an infinite number of $K$ by 1 vectors in the set $\mathbb{Z}^{K \times 1}$, in (\ref{equ:Pw-in-N}), there are a finite number of $K$ by 1 vectors in the set $\mathcal{N}$. This makes it possible to solve problem (\ref{equ:feasibility-problem-of-w-Pw-in-N}) via exhaustive search, as shown in Algorithm \ref{alg:check-whether-Omega-in-Lambda}. 

\begin{algorithm}[t] 
    \caption{Algorithm for Solving Problem (\ref{equ:feasibility-problem-of-w-Pw-in-N})}
    
    \small
    \vspace{3pt}

    \begin{itemize}
    \item[a)] Given some $\boldsymbol{\Omega}$ with $\operatorname{rank}(\boldsymbol{P}(\boldsymbol{\Omega})) = 2$, for each $\boldsymbol{n} \in \mathcal{N}$, solve for $\boldsymbol{w}$ satisfying both $\boldsymbol{P}(\boldsymbol{\Omega})\boldsymbol{w} = \boldsymbol{n}$ and (\ref{equ:constraint-w-2});
    
    \item[b)] If such $\boldsymbol{w}$ is obtained for at least one $\boldsymbol{n} \in \mathcal{N}$, conclude that problem (\ref{equ:feasibility-problem-of-w-Pw-in-N}) is feasible; otherwise, if no such $\boldsymbol{w}$ exists for all $\boldsymbol{n} \in \mathcal{N}$, conclude that problem (\ref{equ:feasibility-problem-of-w-Pw-in-N}) is infeasible.
    \end{itemize}

    \vspace{3pt}
  \label{alg:check-whether-Omega-in-Lambda}
\end{algorithm}

To summarize, under the case when $K>3$, we have the following theorem.

\begin{theorem} \label{thm:K-greater-than-3}
In the case when $K>3$ antennas are selected from $M>3$ antennas, the set $\boldsymbol{\Lambda}$ in (\ref{equ:theorem1}) is characterized as
\begin{align}  \label{equ:Lambda-K-greater-than-3}
\boldsymbol{\Lambda} = \{ \boldsymbol{\Omega} \mid & \operatorname{rank}(\boldsymbol{P}(\boldsymbol{\Omega})) = 2 \mathrm{\ and\ Algorithm\ \ref{alg:check-whether-Omega-in-Lambda}\ claims\ } \nonumber \\ 
&\mathrm{that\ problem\ (\ref{equ:feasibility-problem-of-w-Pw-in-N})\ is\ infeasible} \}.
\end{align}
\end{theorem}

\section{Antenna Selection And Permutation Strategy Design For Eve's Estimation Ambiguity} \label{sec:antenna-selection-strategy-design-ambiguity}

In this section, we first reveal that any antenna selection and permutation  strategy that satisfies Bob's estimation uniqueness requirement as in Definition \ref{def:Bob-uniqueness}, i.e., $\boldsymbol{\Omega} \in \boldsymbol{\Lambda}$, inherently satisfies Eve's estimation ambiguity requirement as in Definition \ref{def:Eve-ambiguity}. 
Then, we will show how to design the antenna selection and permutation strategy to confuse Eve's AOD estimation to a higher level. 

\subsection{Eve's Ambiguity Condition is Always Satisfied}

First, we analyze the property of all the antenna selection and permutation strategies in $\boldsymbol{\Lambda}$ that prevent Eve from knowing ${\theta}_{\mathrm{Eve}}$ and ${\phi}_{\mathrm{Eve}}$ as in Definition \ref{def:Eve-ambiguity} in the following proposition.

\begin{proposition} \label{prop:Eve-ambiguity}
    In the regime of $G \geq K$, i.e., the pilot sequence length is no smaller than the number of selected antennas, an antenna selection and permutation strategy $\boldsymbol{\Omega} \in \boldsymbol{\Lambda}$ can satisfy (\ref{equ:requirement-def2}) in Definition \ref{def:Eve-ambiguity} if and only if 
    given any ${u}_{\mathrm{Eve}}$ and ${v}_{\mathrm{Eve}}$ with $0 < u_{\mathrm{Eve}}^2+v_{\mathrm{Eve}}^2 < 1$, there always exists $\tilde{\boldsymbol{\Omega}} \in \boldsymbol{\Lambda}$, $\tilde{u}_{\mathrm{Eve}}$ and $\tilde{v}_{\mathrm{Eve}}$ with $0<\tilde{u}_{\mathrm{Eve}}^2 + \tilde{v}_{\mathrm{Eve}}^2 < 1$ and $(\tilde{u}_{\mathrm{Eve}}, \tilde{v}_{\mathrm{Eve}})\neq ({u}_{\mathrm{Eve}}, {v}_{\mathrm{Eve}})$ such that all the elements in the vector $\tilde{u}_{\mathrm{Eve}}  \tilde{\boldsymbol{x}}(\tilde{\boldsymbol{\Omega}})  + \tilde{v}_{\mathrm{Eve}} \tilde{\boldsymbol{z}}(\tilde{\boldsymbol{\Omega}})  - {u}_{\mathrm{Eve}} {\boldsymbol{x}}({\boldsymbol{\Omega}})  - {v}_{\mathrm{Eve}} {\boldsymbol{z}}({\boldsymbol{\Omega}}) $ are even integers.
\end{proposition}

\begin{IEEEproof}
See Appendix \ref{app:proof-theorem-Eve}.
\end{IEEEproof}

Intuitively, if Eve can find $\tilde{u}_{\mathrm{Eve}}$, $\tilde{v}_{\mathrm{Eve}}$, and $\tilde{\boldsymbol{\Omega}}$ such that all the elements of $\tilde{u}_{\mathrm{Eve}}  \tilde{\boldsymbol{x}}(\tilde{\boldsymbol{\Omega}})  + \tilde{v}_{\mathrm{Eve}} \tilde{\boldsymbol{z}}(\tilde{\boldsymbol{\Omega}})  - {u}_{\mathrm{Eve}} {\boldsymbol{x}}({\boldsymbol{\Omega}})  - {v}_{\mathrm{Eve}} {\boldsymbol{z}}({\boldsymbol{\Omega}}) $ are even integers, this new solution will add a phase difference of $2\pi$ to all the elements in the channel $\tilde{\boldsymbol{h}}_{\mathrm{Eve}}(\boldsymbol{\Omega}, u_{\mathrm{Eve}}, v_{\mathrm{Eve}})$ given in (\ref{equ:Eve-tilde-h-B-u-v}). This will lead to the same received signals at Eve and make Eve have ambiguity in AOD estimation. 
Notably, we show in the following that any antenna strategy $\boldsymbol{\Omega} \in \boldsymbol{\Lambda}$ necessarily induces such estimation ambiguity at Eve.

\begin{theorem} \label{thm:eight-ambiguity}
    Any antenna selection and permutation strategy $\boldsymbol{\Omega} \in \boldsymbol{\Lambda}$, where $\boldsymbol{\Lambda}$ is characterized by (\ref{equ:Lambda-K-equal-3}) when $K=3$ as in Theorem \ref{thm:equ-thm-uniqueness-K-equal-3} and by (\ref{equ:Lambda-K-greater-than-3}) when $K>3$ as in Theorem \ref{thm:K-greater-than-3}, can cause AOD estimation ambiguity to Eve as in Definition \ref{def:Eve-ambiguity}.    
    Specifically, given any AOD from Alice to Eve as well as antenna strategy satisfying the following conditions:
    \begin{enumerate}
        \item[a)] $u_{\mathrm{Eve}}$ and $v_{\mathrm{Eve}}$ satisfy $u_{\mathrm{Eve}} \neq 0$ and $v_{\mathrm{Eve}} \neq 0$;
        \item[b)] $u_{\mathrm{Eve}}$ and $v_{\mathrm{Eve}}$ satisfy $|u_{\mathrm{Eve}}| \neq |v_{\mathrm{Eve}}|$;
        \item[c)] $\boldsymbol{x}(\boldsymbol{\Omega})$ and $\boldsymbol{z}(\boldsymbol{\Omega})$ satisfy $\max\{\boldsymbol{x}(\boldsymbol{\Omega})\} - \min\{\boldsymbol{x}(\boldsymbol{\Omega})\} < M_z$ and $\max\{\boldsymbol{z}(\boldsymbol{\Omega})\} - \min\{\boldsymbol{z}(\boldsymbol{\Omega})\} < M_x$,
    \end{enumerate}
    Eve, with $\boldsymbol{\Omega}$ unknown, can estimate at least eight distinct AODs that lead to its received signal, which include:
    {\allowdisplaybreaks
    \begin{subequations} \label{eq:aod_8_all}
        \begin{align}
            (\tilde{\theta}_{\mathrm{Eve}}, \tilde{\phi}_{\mathrm{Eve}}) &= (\theta_{\mathrm{Eve}}, \phi_{\mathrm{Eve}}), \label{eq:aod_8_1} \\
            (\tilde{\theta}_{\mathrm{Eve}}, \tilde{\phi}_{\mathrm{Eve}}) &= (\theta_{\mathrm{Eve}}, \pi-\phi_{\mathrm{Eve}}), \label{eq:aod_8_2} \\
            (\tilde{\theta}_{\mathrm{Eve}}, \tilde{\phi}_{\mathrm{Eve}}) &= (-\theta_{\mathrm{Eve}}, \phi_{\mathrm{Eve}}), \label{eq:aod_8_3} \\
            (\tilde{\theta}_{\mathrm{Eve}}, \tilde{\phi}_{\mathrm{Eve}}) &= (-\theta_{\mathrm{Eve}}, \pi-\phi_{\mathrm{Eve}}), \label{eq:aod_8_4} \\
            (\tilde{\theta}_{\mathrm{Eve}}, \tilde{\phi}_{\mathrm{Eve}}) &= (\arcsin u_{\mathrm{Eve}}, \arccos\tfrac{v_{\mathrm{Eve}} }{ \sqrt{1-u_{\mathrm{Eve}}^2}}), \label{eq:aod_8_5} \\
            (\tilde{\theta}_{\mathrm{Eve}}, \tilde{\phi}_{\mathrm{Eve}}) &= (\arcsin u_{\mathrm{Eve}}, \pi - \arccos \tfrac{v_{\mathrm{Eve}} }{ \sqrt{1-u_{\mathrm{Eve}}^2}}), \label{eq:aod_8_6} \\
            (\tilde{\theta}_{\mathrm{Eve}}, \tilde{\phi}_{\mathrm{Eve}}) &= (-\arcsin u_{\mathrm{Eve}}, \arccos \tfrac{v_{\mathrm{Eve}} }{ \sqrt{1-u_{\mathrm{Eve}}^2}}), \label{eq:aod_8_7} \\
            (\tilde{\theta}_{\mathrm{Eve}}, \tilde{\phi}_{\mathrm{Eve}}) &= (-\arcsin u_{\mathrm{Eve}}, \pi - \arccos \tfrac{v_{\mathrm{Eve}}} { \sqrt{1-u_{\mathrm{Eve}}^2}}). \label{eq:aod_8_8}
        \end{align}
    \end{subequations}
    }%
    If condition c) is satisfied and $u_{\mathrm{Eve}} = 0$, $v_{\mathrm{Eve}} \neq 0$, Eve can estimate at least four distinct AODs as follows:
    \begin{subequations} \label{eq:aod_4u}
    \begin{align}
        (\tilde{\theta}_{\mathrm{Eve}}, \tilde{\phi}_{\mathrm{Eve}}) &= (\theta_{\mathrm{Eve}}, \pi/2), \label{eq:aod_4u_1} \\
        (\tilde{\theta}_{\mathrm{Eve}}, \tilde{\phi}_{\mathrm{Eve}}) &= (-\theta_{\mathrm{Eve}}, \pi/2), \label{eq:aod_4u_2} \\
        (\tilde{\theta}_{\mathrm{Eve}}, \tilde{\phi}_{\mathrm{Eve}}) &= (0, \arccos v_{\mathrm{Eve}}), \label{eq:aod_4u_3} \\
        (\tilde{\theta}_{\mathrm{Eve}}, \tilde{\phi}_{\mathrm{Eve}}) &= (0, \pi - \arccos v_{\mathrm{Eve}}). \label{eq:aod_4u_4}
    \end{align}
    \end{subequations}
    If condition c) is satisfied and $u_{\mathrm{Eve}} \neq 0$, $v_{\mathrm{Eve}} = 0$, Eve can estimate at least four distinct AODs as follows:
    \begin{subequations} \label{eq:aod_4v}
    \begin{align}
        (\tilde{\theta}_{\mathrm{Eve}}, \tilde{\phi}_{\mathrm{Eve}}) &= (0, \phi_{\mathrm{Eve}}), \label{eq:aod_4v_1} \\
        (\tilde{\theta}_{\mathrm{Eve}}, \tilde{\phi}_{\mathrm{Eve}}) &= (0, \pi - \phi_{\mathrm{Eve}}), \label{eq:aod_4v_2} \\
        (\tilde{\theta}_{\mathrm{Eve}}, \tilde{\phi}_{\mathrm{Eve}}) &= (\arcsin u_{\mathrm{Eve}}, \pi/2), \label{eq:aod_4v_3} \\
        (\tilde{\theta}_{\mathrm{Eve}}, \tilde{\phi}_{\mathrm{Eve}}) &= (-\arcsin u_{\mathrm{Eve}}, \pi/2). \label{eq:aod_4v_4}
    \end{align}
    \end{subequations}
    If condition c) is satisfied and $|u_{\mathrm{Eve}}| = |v_{\mathrm{Eve}}|$, the set of AODs given by (\ref{eq:aod_8_5}) to (\ref{eq:aod_8_8}) is equivalent to the set given by (\ref{eq:aod_8_1}) to (\ref{eq:aod_8_4}), and thus Eve can estimate at least four distinct AODs given by (\ref{eq:aod_8_1}) to (\ref{eq:aod_8_4}).
    Next, if condition c) is violated but condition a) is satisfied, no matter whether condition b) is satisfied, Eve can also estimate at least four distinct AODs given by (\ref{eq:aod_8_1}) to (\ref{eq:aod_8_4}). 
    Moreover, if condition c) is violated and $u_{\mathrm{Eve}} = 0$, $v_{\mathrm{Eve}} \neq 0$, Eve can estimate at least two distinct AODs given by (\ref{eq:aod_4u_1}) and (\ref{eq:aod_4u_2}).
    Finally, if condition c) is violated and $u_{\mathrm{Eve}} \neq 0$, $v_{\mathrm{Eve}} = 0$, Eve can estimate at least two distinct AODs given by (\ref{eq:aod_4v_1}) and (\ref{eq:aod_4v_2}).
\end{theorem}
\begin{IEEEproof}
    See Appendix \ref{app:eight-ambiguity}.
\end{IEEEproof}

In the above, Conditions a) and b) simply require that Eve does not stand on specific planes $x=0$, $z=0$ and $x=\pm z$ in the whole space, while Condition c) constrains the maximum distance between selected antennas along one axis to the total length of the other axis, tricking Eve into interpreting the $x$-axis phase differences as the $z$-axis ones, and vice versa.
Moreover, the above eight AOD estimations correspond to eight symmetric propagation directions in the space, which are physically related through reflections with respect to the planes $x=0$, $z=0$ and $x=\pm z$. Regardless of how many distinct AODs remain from these eight estimations, the resulting AOD solutions constitute a \textit{basic solution set} that can always confuse Eve as in Definition \ref{def:Eve-ambiguity}. 

\subsection{Enhanced Ambiguity} \label{subsec:enhanced-ambiguity}

We have shown that given any antenna selection and permutation strategy $\boldsymbol{\Omega} \in \boldsymbol{\Lambda}$, Eve gets confused among two, four, or eight AOD solutions from the basic solution set. Nevertheless, a larger number of AOD estimations at Eve is highly desirable. For example, more estimations impose higher cost for Eve to launch targeted attacks against all plausible candidate directions. 
Because many antenna selection and permutation strategies in $\boldsymbol{\Lambda}$ achieve the above goal, we are interested in the following enhanced privacy problem. Can we find a set $\tilde{\boldsymbol{\Lambda}} \subset \boldsymbol{\Lambda}$ such that even if this set is known by Eve, as long as Alice's strategy $\boldsymbol{\Omega}$ is limited to the set $\tilde{\boldsymbol{\Lambda}}$, Eve will always be confused by more AOD solutions beyond the basic solution set with two, four, or eight AODs as shown in Theorem 3?

To this end, we first note that each AOD solution within the basic solution set is paired with a corresponding antenna strategy to yield the same received signal at Eve. For example, AOD in (\ref{eq:aod_8_1}) obviously corresponds to the true strategy $\boldsymbol{\Omega}$. These strategies constitute a \textit{basic strategy group}, which must be included in $\tilde{\boldsymbol{\Lambda}}$ to ensure that Eve is at least confused by the basic solution set. 
Next, suppose Eve can estimate an AOD outside the basic solution set that also leads to its received signal, under a specific strategy in $\tilde{\boldsymbol{\Lambda}}$. Then, it is guaranteed that this additional AOD estimation can also expand into an \textit{enhanced solution set} with two, four, or eight distinct AODs, analogous to the basic solution set shown in Theorem \ref{thm:eight-ambiguity}. Eve will also be confused by this enhanced solution set, as long as its corresponding \textit{enhanced strategy group} is included in $\tilde{\boldsymbol{\Lambda}}$. 
Therefore, a higher privacy level can be achieved by including the basic strategy group and several enhanced strategy groups in $\tilde{\boldsymbol{\Lambda}}$, such that Eve will always be confused by multiple AOD solution sets, including the basic one and some enhanced ones. 
In the following, we define more precisely how $\tilde{\boldsymbol{\Lambda}}$ achieves a higher level of privacy.

\begin{definition}[Enhanced Ambiguity in Eve's AOD Estimation] \label{def:Q-level-privacy}
    A set of antenna selection and permutation strategies $\tilde{\boldsymbol{\Lambda}}$ achieves $Q$-level privacy if given any $\boldsymbol{\Omega} \in \tilde{\boldsymbol{\Lambda}}$ and any AOD from Alice to Eve, $\tilde{\boldsymbol{\Lambda}}$ contains at least $Q$ strategy groups (one basic including $\boldsymbol{\Omega}$ and at least $Q-1$ enhanced), which, together with their corresponding AOD solution sets, all yield the same received signal at Eve.
\end{definition}

While Definition \ref{def:Q-level-privacy} provides the requirement of $\tilde{\boldsymbol{\Lambda}}$ to achieve $Q$-level privacy, constructing such a set remains highly challenging. As with Definition \ref{def:Eve-ambiguity}, this requirement must be satisfied for every possible AOD from Alice to Eve. However, the AOD varies continuously, making exhaustive verification intractable.
To tackle the above challenge, we reveal a hidden mechanism behind the existence of multiple AOD estimations at Eve, which is embedded in Proposition \ref{prop:Eve-ambiguity}, but not fully recognized in the previous analysis.
Specifically, Proposition \ref{prop:Eve-ambiguity} provides a condition on $K$-dimensional vectors $u_{\mathrm{Eve}} {\boldsymbol{x}}({\boldsymbol{\Omega}}) + v_{\mathrm{Eve}} {\boldsymbol{z}}({\boldsymbol{\Omega}})$ and $\tilde{u}_{\mathrm{Eve}} \tilde{\boldsymbol{x}}(\tilde{\boldsymbol{\Omega}}) + \tilde{v}_{\mathrm{Eve}} \tilde{\boldsymbol{z}}(\tilde{\boldsymbol{\Omega}})$ to satisfy Definition \ref{def:Eve-ambiguity}. An interesting observation is that these $K$-dimensional vectors actually reside in two-dimensional (2D) subspaces of $\mathbb{R}^K$, which are respectively spanned by spatial bases $\{\boldsymbol{x}(\boldsymbol{\Omega}), \boldsymbol{z}(\boldsymbol{\Omega})\}$ and $\{\tilde{\boldsymbol{x}}(\tilde{\boldsymbol{\Omega}}), \tilde{\boldsymbol{z}}(\tilde{\boldsymbol{\Omega}})\}$, while their continuous variation within these subspaces is dictated by 2D angular parameters $(u_{\mathrm{Eve}}, v_{\mathrm{Eve}})$ and $(\tilde{u}_{\mathrm{Eve}}, \tilde{v}_{\mathrm{Eve}})$, respectively.

Intuitively, for a vector $u_{\mathrm{Eve}} {\boldsymbol{x}}({\boldsymbol{\Omega}}) + v_{\mathrm{Eve}} {\boldsymbol{z}}({\boldsymbol{\Omega}})$ in one 2D subspace, it is not easy to find a vector $\tilde{u}_{\mathrm{Eve}} \tilde{\boldsymbol{x}}(\tilde{\boldsymbol{\Omega}}) + \tilde{v}_{\mathrm{Eve}} \tilde{\boldsymbol{z}}(\tilde{\boldsymbol{\Omega}})$ in another 2D subspace such that $\tilde{u}_{\mathrm{Eve}} \tilde{\boldsymbol{x}}(\tilde{\boldsymbol{\Omega}}) + \tilde{v}_{\mathrm{Eve}} \tilde{\boldsymbol{z}}(\tilde{\boldsymbol{\Omega}}) - (u_{\mathrm{Eve}} {\boldsymbol{x}}({\boldsymbol{\Omega}}) + v_{\mathrm{Eve}} {\boldsymbol{z}}({\boldsymbol{\Omega}})) \in 2\mathbb{Z}^K
$ as in Proposition \ref{prop:Eve-ambiguity}, especially under the continuous variation of $(u_{\mathrm{Eve}}, v_{\mathrm{Eve}})$. 
However, if these two vectors can be made to lie in the same 2D subspace, the original $K$-dimensional matching condition will simplify to a 2D mapping between $(u_{\mathrm{Eve}},v_{\mathrm{Eve}})$ and $(\tilde{u}_{\mathrm{Eve}},\tilde{v}_{\mathrm{Eve}})$. To hold for all $(u_{\mathrm{Eve}},v_{\mathrm{Eve}})$, this pointwise condition naturally becomes a 2D region-covering requirement. Excitingly, based on this insight, we manage to find a region-covering structure in the 2D plane that is sufficient to construct a set $\tilde{\boldsymbol{\Lambda}}$ satisfying Definition \ref{def:Q-level-privacy}, as shown in the following theorem.

\begin{theorem}\label{thm:Q-level-construction}
Suppose there exists a set $\mathcal{T} = \{\boldsymbol{T}_1, \dots, \boldsymbol{T}_P\}$ with $\boldsymbol{T}_1 = \boldsymbol{I}_2$ and $\boldsymbol{T}_p \in \mathrm{GL}(2, \mathbb{Z})$, $p = 2, \dots, P$, such that 
\begin{equation} 
\boldsymbol{T}_{p_1}^{-1}\boldsymbol{T}_{p_2} \notin \boldsymbol{\Sigma}, \ \forall 1\le p_1 < p_2 \le P, 
\label{equ:thm-square-non-overlap}
\end{equation}
and the following region-covering condition holds:
\begin{equation} \label{equ:thm-square-covering}
    \sum_{\boldsymbol{T}_p \in \mathcal{T}} \mathbbold{1}_{\mathcal{E}_0(\boldsymbol{T}_p) + 2\mathbb{Z}^2}(\boldsymbol{w}) \geq Q, \  \forall \boldsymbol{w} \in [-1, 1]^2 \setminus \left(\mathbb{Z}^2 \setminus 2\mathbb{Z}^2\right),
\end{equation}
where $\boldsymbol{\Sigma} = \{\boldsymbol{\Pi} \in \mathbb{Z}^{2 \times 2} \mid \boldsymbol{\Pi}\boldsymbol{\Pi}^T = \boldsymbol{I}_2\}$ denotes the set of all $2\times2$ signed permutation matrices, $\mathbbold{1}_{\mathcal{I}}$ is the indicator function for set $\mathcal{I}$, and $\mathcal{E}_0(\boldsymbol{T}_p) = \{\boldsymbol{w} \in \mathbb{R}^2 \mid \|\boldsymbol{T}_p^{-1}\boldsymbol{w}\|_2 < 1\}$ denotes the basic elliptic region associated with $\boldsymbol{T}_p$.
Then, define $\mathcal{R}$ as the set of all strategies $\boldsymbol{\Omega} \in \boldsymbol{\Lambda}$ such that for every $\boldsymbol{T}_p \in \mathcal{T}$, at least one of the following two constraint sets is satisfied:
\begin{equation} \label{equ:thm-physical-bound}
    \left\{\begin{array}{l}
    \max\{\boldsymbol{P}(\boldsymbol{\Omega})[\boldsymbol{T}_p]_{:,1}\} - \min\{\boldsymbol{P}(\boldsymbol{\Omega})[\boldsymbol{T}_p]_{:,1}\} < M_x, \\
    \max\{\boldsymbol{P}(\boldsymbol{\Omega})[\boldsymbol{T}_p]_{:,2}\} - \min\{\boldsymbol{P}(\boldsymbol{\Omega})[\boldsymbol{T}_p]_{:,2}\} < M_z.
    \end{array}\right.
\end{equation}
\begin{equation} \label{equ:thm-physical-bound-2}
    \left\{\begin{array}{l}
    \max\{\boldsymbol{P}(\boldsymbol{\Omega})[\boldsymbol{T}_p]_{:,1}\} - \min\{\boldsymbol{P}(\boldsymbol{\Omega})[\boldsymbol{T}_p]_{:,1}\} < M_z, \\
    \max\{\boldsymbol{P}(\boldsymbol{\Omega})[\boldsymbol{T}_p]_{:,2}\} - \min\{\boldsymbol{P}(\boldsymbol{\Omega})[\boldsymbol{T}_p]_{:,2}\} < M_x.
    \end{array}\right.
\end{equation}
where $[\boldsymbol{T}_p]_{:,1}$ and $[\boldsymbol{T}_p]_{:,2}$ denote the first and second columns of $\boldsymbol{T}_p$, respectively.
Finally, a strategy set satisfying Definition \ref{def:Q-level-privacy} can be constructed as
\begin{equation} \label{equ:thm-Lambda-construction}
\tilde{\boldsymbol{\Lambda}} = \bigcup_{\boldsymbol{\Omega} \in \mathcal{R}} \bigcup_{\boldsymbol{T}_p \in \mathcal{T}} \left\{ \tilde{\boldsymbol{\Omega}} \in \boldsymbol{\Lambda} \mid \tilde{\boldsymbol{P}}(\tilde{\boldsymbol{\Omega}}) = \boldsymbol{P}(\boldsymbol{\Omega})\boldsymbol{T}_p\boldsymbol{\Pi}, \ \boldsymbol{\Pi} \in \boldsymbol{\Sigma} \right\},
\end{equation}
where $\tilde{\boldsymbol{P}}(\tilde{\boldsymbol{\Omega}}) = [\tilde{\boldsymbol{x}}(\tilde{\boldsymbol{\Omega}}), \tilde{\boldsymbol{z}}(\tilde{\boldsymbol{\Omega}})] \in \mathbb{Z}^{K \times 2}$.
\end{theorem}

\begin{IEEEproof}
    See Appendix \ref{app:Q-level-construction}.
\end{IEEEproof}

Theorem \ref{thm:Q-level-construction} provides a constructive way to realize the $Q$-level privacy requirement in Definition \ref{def:Q-level-privacy}. Each transformation matrix $\boldsymbol{T}_p \in \mathcal{T}$ generates one physically realizable strategy group, while the signed permutations $\boldsymbol{\Pi} \in \boldsymbol{\Sigma}$ enumerate all valid strategies within that group. Then, $\tilde{\boldsymbol{\Lambda}}$ is constructed in \eqref{equ:thm-Lambda-construction} to simultaneously include all such $P$ groups.
The key idea behind this construction is to distribute the responsibility of confusing Eve among $P$ strategy groups associated with a specific $\boldsymbol{\Omega} \in \tilde{\boldsymbol{\Lambda}}$. Each strategy group is effective in providing ambiguity only over a specific region of the continuous AOD domain, rather than all possible true AODs. Meanwhile, the region-covering condition (\ref{equ:thm-square-covering}) guarantees that every possible AOD is covered by at least $Q$ such regions, so that Eve will always face ambiguity from at least $Q$ strategy groups.

All that remains to construct such $\tilde{\boldsymbol{\Lambda}}$ is to find a set $\mathcal{T}$ that satisfies \eqref{equ:thm-square-covering}, which can be efficiently solved through a combinatorial search over the discrete space $\mathrm{GL}(2,\mathbb{Z})$. Fig. \ref{fig:examples_coverage_requirement_T} shows an example solution of $\mathcal{T}$ that satisfies (\ref{equ:thm-square-covering}) for a target covering multiplicity of $Q = 2$. It is observed in Fig. \ref{fig:examples_coverage_requirement_T}\subref{fig:example1_Q2_coverage} that the covering multiplicity by all the corresponding basic elliptic regions and their translations by $2\mathbb{Z}^2$ is consistently at least two within the region $[-1,1]^2 \backslash\left(\mathbb{Z}^2 \backslash 2 \mathbb{Z}^2\right)$. 
\begin{figure}[t]
\centering
\subfloat[]{
    \includegraphics[trim=0.45cm 0cm 1cm 1.2cm, clip, width = 0.3\linewidth]{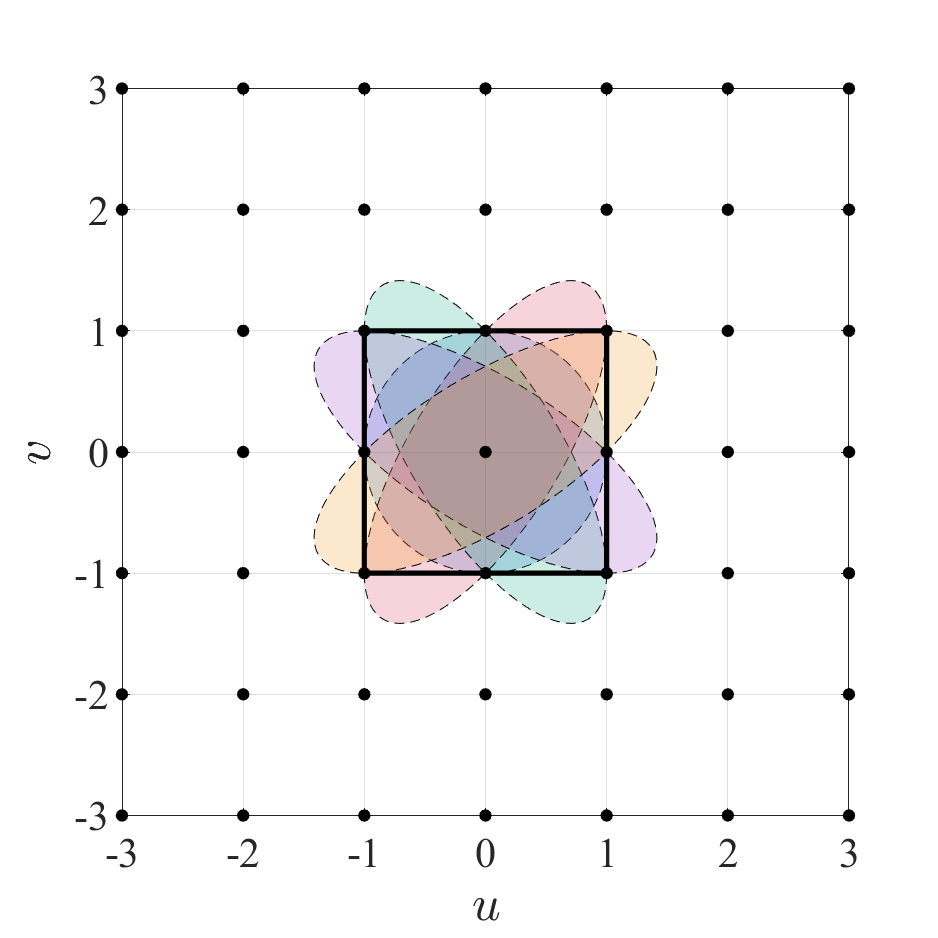}
    \label{fig:example1_Q2_pattern_k0}
}
\subfloat[]{
    \includegraphics[trim=0.45cm 0cm 1cm 1.2cm, clip, width = 0.3\linewidth]{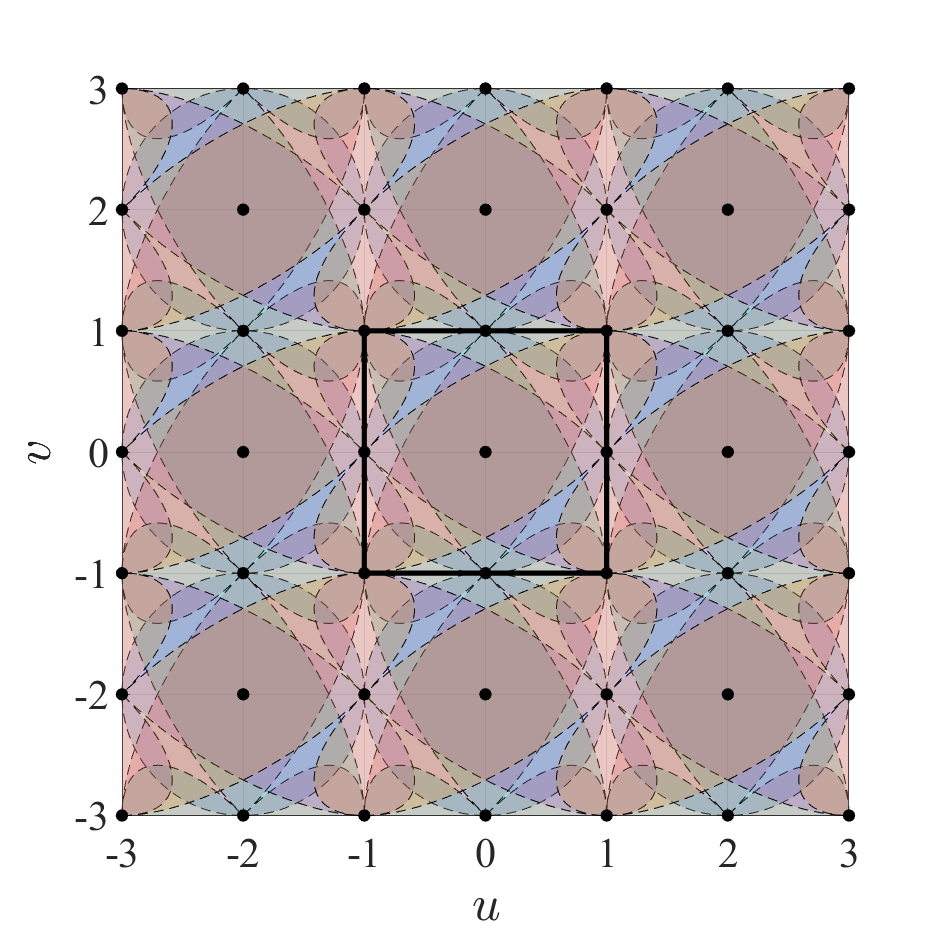}
    \label{fig:example1_Q2_pattern}
}
\subfloat[]{
    \includegraphics[trim=0.35cm 1.3cm 0.7cm 2.2cm, clip, width = 0.365\linewidth]{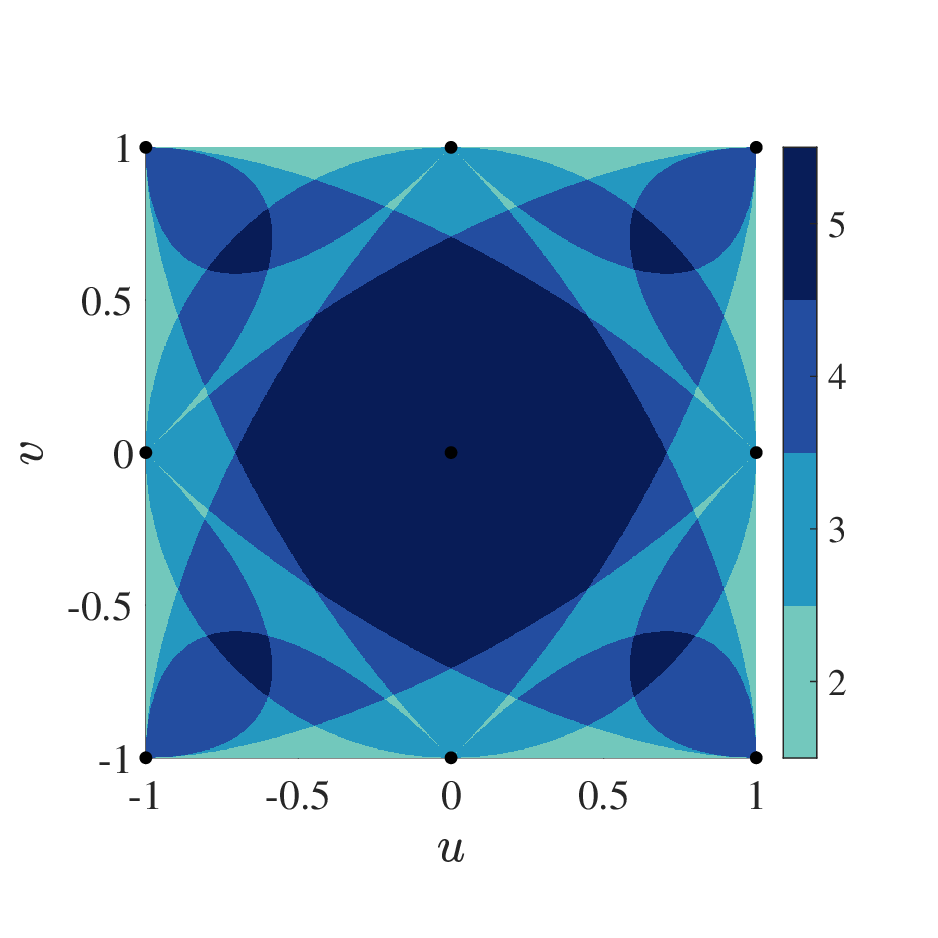}
    \label{fig:example1_Q2_coverage}
}
\caption{Illustration of a solution set that satisfies \eqref{equ:thm-square-covering} for $Q = 2$: $\mathcal{T} = \big\{ \big[\begin{smallmatrix} 1 & 0 \\ 0 & 1 \end{smallmatrix}\big], \big[\begin{smallmatrix} -1 & 0 \\ -1 & -1 \end{smallmatrix}\big], \big[\begin{smallmatrix} 0 & 1 \\ -1 & -1 \end{smallmatrix}\big], \big[\begin{smallmatrix} 1 & 1 \\ -1 & 0 \end{smallmatrix}\big], \big[\begin{smallmatrix} -1 & 1 \\ -1 & 0 \end{smallmatrix}\big] \big\}$. (a) The corresponding basic elliptic regions $\mathcal{E}_0(\boldsymbol{T}_p)$, (b) Their translations by $2\mathbb{Z}^2$, (c) The coverage multiplicity map of $[-1,1]^2 \backslash (\mathbb{Z}^2 \backslash 2 \mathbb{Z}^2)$. The region $[-1,1]^2$ is framed by a solid black box in (a) and (b). 
}
\label{fig:examples_coverage_requirement_T}
\end{figure}

\begin{remark}
Note that $\tilde{\boldsymbol{\Lambda}}$ in (\ref{equ:thm-Lambda-construction}) becomes empty if $\mathcal{R}=\emptyset$. This may occur when the target $Q$ is set too high for a given UPA size $M_x \times M_z$. In particular, achieving a higher privacy level $Q$ requires a larger set $\mathcal{T}$, making it  difficult to find $\boldsymbol{\Omega}$ that consistently satisfies at least one of the constraint sets in (\ref{equ:thm-physical-bound}) and (\ref{equ:thm-physical-bound-2}) for all required spatial transformations $\boldsymbol{T}_p$.
\end{remark}

\section{Numerical Results}  \label{sec:numerical-results}

In this section, we provide numerical examples to verify the effectiveness of our proposed transmit antenna selection and permutation strategy in privacy-preserving localization.

\subsection{Bob's Uniqueness and Eve's Ambiguity}
% 1. noiseless. can you start with a small array. for example, choose 4 from 8? make sure Lambda consists a small number of sets such that you can list all the solutions. show Fig. 6. you can show some Omega and AOD such that Eve has 2 solutions, 8 solutions, and 16 solutions. for example, you can also show with the case of 8 solutions, what Omega lead to these solutions

% you should also discuss how K and M affect the performance 

% 2. noiseless case

% other than the basic result, you should also discuss how different setups affect the performance

First, we consider a UPA with $M_x=4$, $M_z=2$, and aim to select and permute $K = 4$ antennas out of all these $M = 8$ antennas for transmission. The set of all the feasible antenna selection and permutation strategies $\boldsymbol{\Lambda}$ is obtained by applying the algorithm proposed in Section \ref{sec:antenna-selection-strategy-design-uniqueness}. For clarity, we only illustrate the resulting feasible combinations of the selected $K = 4$ antennas in Fig. \ref{fig:all-combinations}, where the selected antennas are represented by solid circles. The complete set $\boldsymbol{\Lambda}$ is then formed by considering all $K!$ permutations for each feasible combination.
Then, Alice can arbitrarily choose a strategy from $\boldsymbol{\Lambda}$ to transmit pilot signals. In the following, we show the performance in both noiseless and noisy cases.

\begin{figure*}[t]
    \centering
        \includegraphics[width=0.98\linewidth]{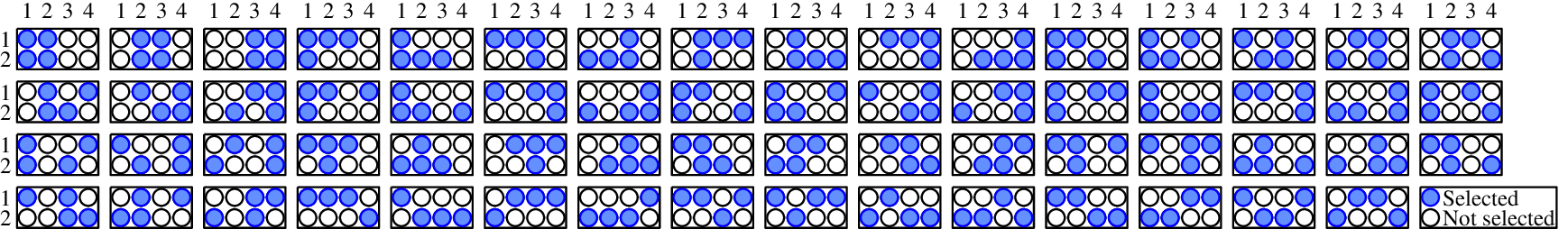}
    \caption{All feasible combinations of the selected $K=4$ antennas from a $4\times 2$ UPA to form the strategy set $\boldsymbol{\Lambda}$.} 
    \label{fig:all-combinations}
    \vspace{-3.0mm}
\end{figure*}
\begin{figure*}[t]
    \centering
    \subfloat[]{
        \includegraphics[trim=0cm 0cm 0.8cm 1.0cm, clip, width=0.18\linewidth]{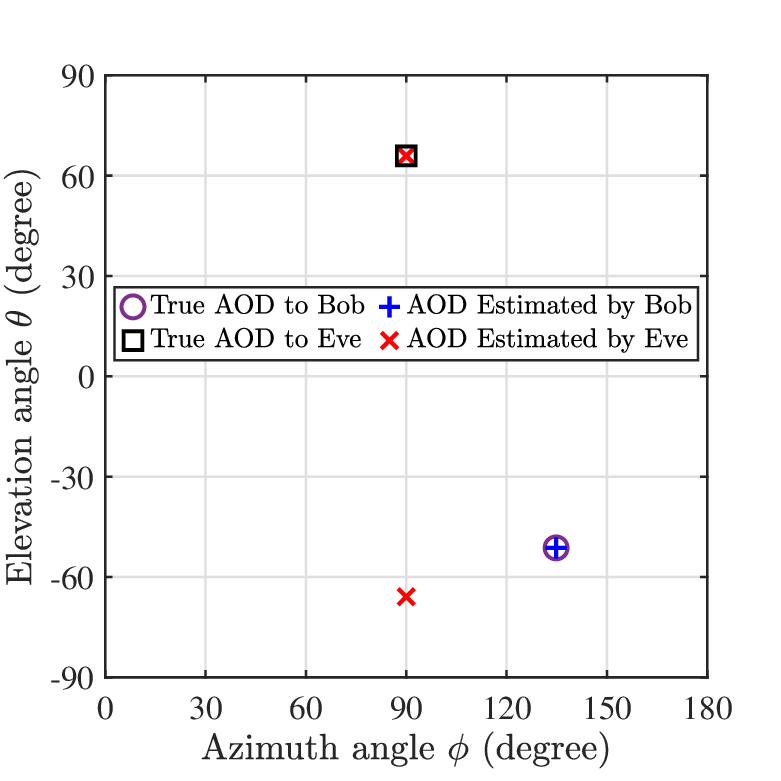} 
        \label{fig:u0-success-2-solutions}
    }
    \hfill
    \subfloat[]{
        \includegraphics[trim=0cm 0cm 0.8cm 1.0cm, clip, width=0.18\linewidth]{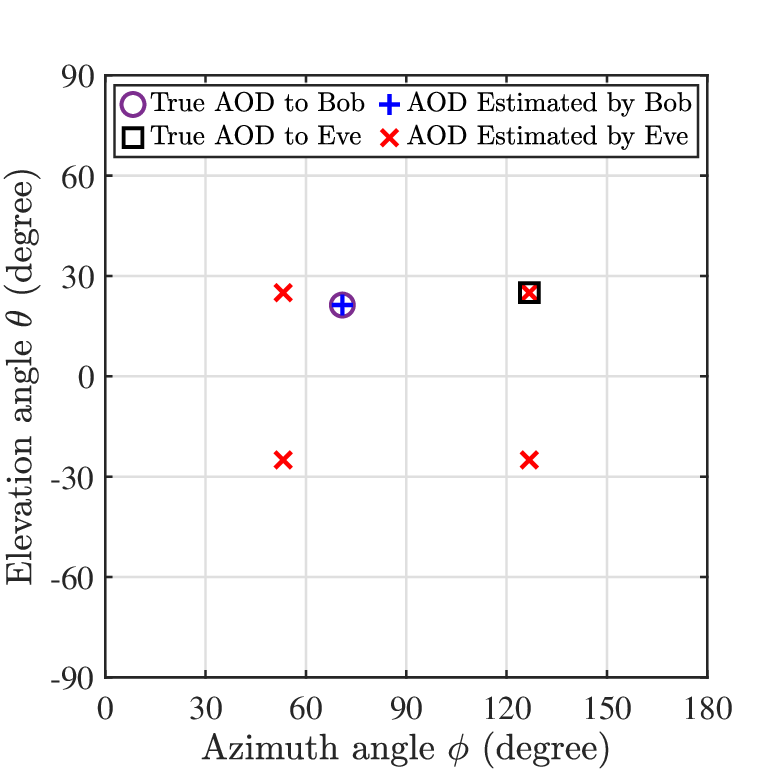}
        \label{fig:4-AOD-solutions}
    }
    \hfill
    \subfloat[]{
        \includegraphics[trim=0cm 0cm 0.8cm 1.0cm, clip, width=0.18\linewidth]{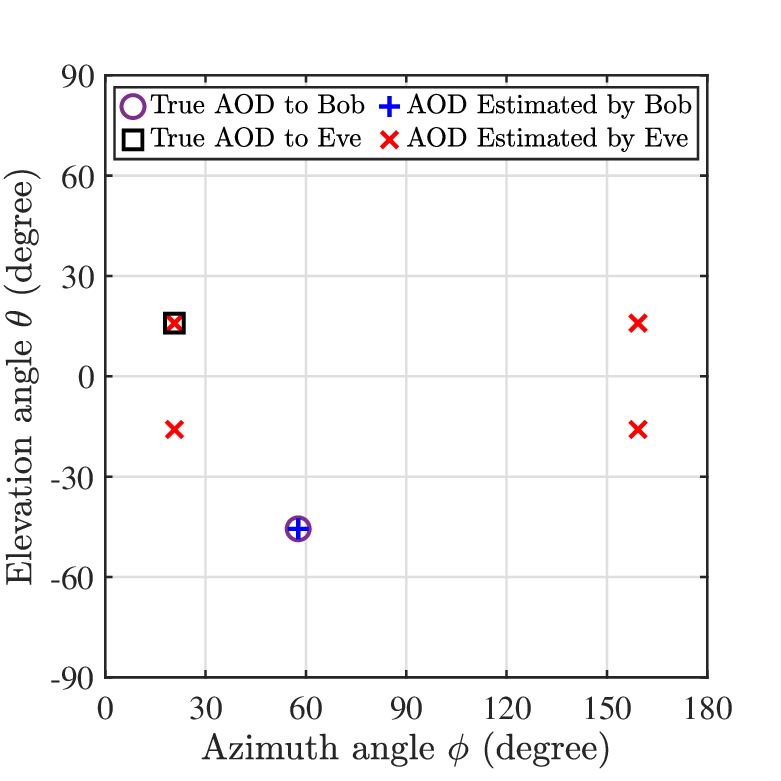} %0.93/
        \label{fig:4-AOD-solutions-2}
    } 
    \hfill
    \subfloat[]{
        \includegraphics[trim=0cm 0cm 0.8cm 1.0cm, clip, width=0.18\linewidth]{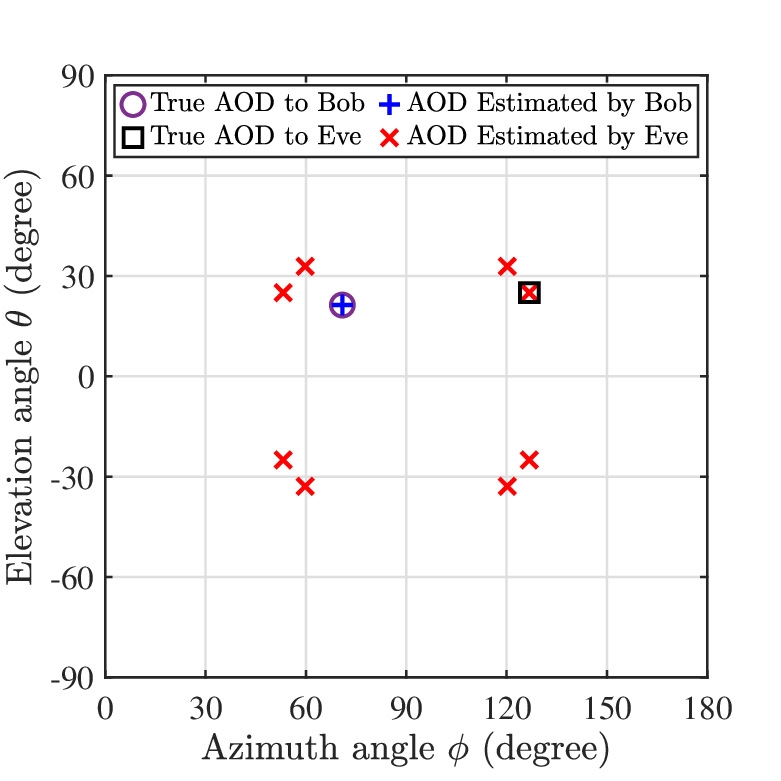}
        \label{fig:8-AOD-solutions}
    }
    \hfill
    \subfloat[]{
        \includegraphics[trim=0cm 0cm 0.8cm 1.0cm, clip, width=0.18\linewidth]{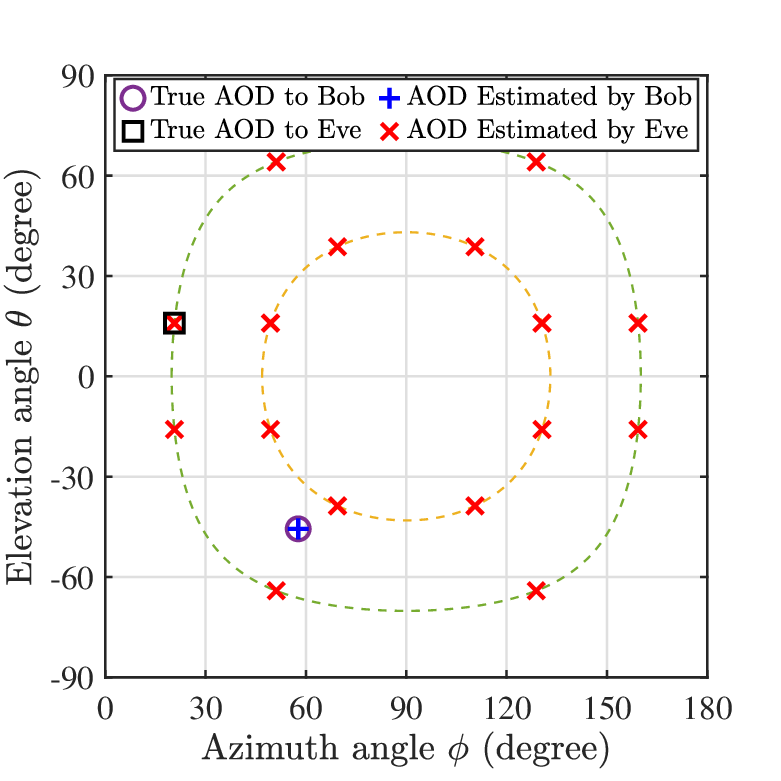}
        \label{fig:16-AOD-solutions}
    }
    \caption{AOD estimation by Bob and Eve in the noiseless case:  
    (a) $\boldsymbol{\Omega} = ((1,1),(2,1),(3,1),(1,2))$, $\theta_{\mathrm{Bob}}=-67.7^\circ$, $\phi_{\mathrm{Bob}}=153.6^\circ$, $\theta_{\mathrm{Eve}}\ = 65.9^\circ$, $\phi_{\mathrm{Eve}} = 90^\circ$;
    (b) $\boldsymbol{\Omega} = ((1,1),(4,1),(1,2),(3,2))$, $\theta_{\mathrm{Bob}}=21.3^\circ$, $\phi_{\mathrm{Bob}}=70.9^\circ$, $\theta_{\mathrm{Eve}}\ = 25.0^\circ$, $\phi_{\mathrm{Eve}} = 126.8^\circ$;
    (c) $\boldsymbol{\Omega} = ((1,1),(4,1),(1,2),(3,2))$, $\theta_{\mathrm{Bob}}=-45.6^\circ$, $\phi_{\mathrm{Bob}}=57.7^\circ$, $\theta_{\mathrm{Eve}}\ = 15.9^\circ$, $\phi_{\mathrm{Eve}} = 20.7^\circ$;
    (d) and (e) use the same $\boldsymbol{\Omega}$ and true AODs to Bob and Eve as (b) and (c), respectively, but consider a $4 \times 4$ UPA.
    }
    \vspace{-3mm}
    \label{fig:noiseless-case}
\end{figure*}
\begin{figure*}[t]
    \centering
    \subfloat[]{
        \includegraphics[width=0.10\linewidth]{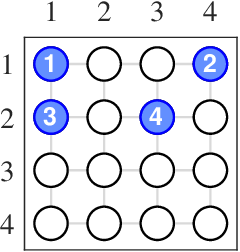} 
        \label{fig:30a-Omega}
    }
    \hfill
    \subfloat[]{
        \includegraphics[width=0.10\linewidth]{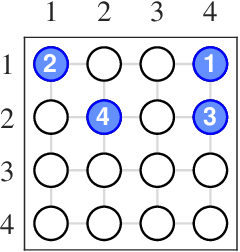}
        \label{fig:30b-Omega}
    }
    \hfill
    \subfloat[]{
        \includegraphics[width=0.10\linewidth]{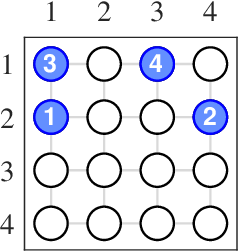}
        \label{fig:30c-Omega}
    }
    \hfill
    \subfloat[]{
        \includegraphics[width=0.10\linewidth]{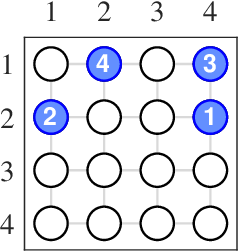} 
        \label{fig:30d-Omega}
    }
    \hfill
    \subfloat[]{
        \includegraphics[width=0.10\linewidth]{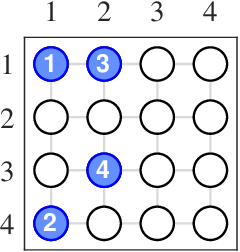} 
        \label{fig:30e-Omega}
    }
    \hfill
    \subfloat[]{
        \includegraphics[width=0.10\linewidth]{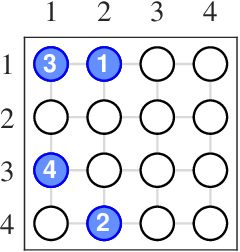}
        \label{fig:30f-Omega}
    }
    \hfill
    \subfloat[]{
        \includegraphics[width=0.10\linewidth]{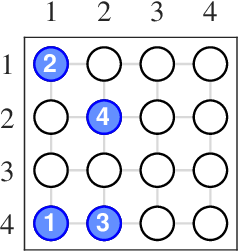}
        \label{fig:30g-Omega}
    }
    \hfill
    \subfloat[]{
        \includegraphics[width=0.10\linewidth]{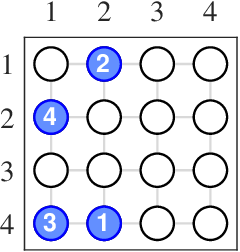} 
        \label{fig:30h-Omega}
    } \\
    \subfloat[]{
        \includegraphics[width=0.10\linewidth]{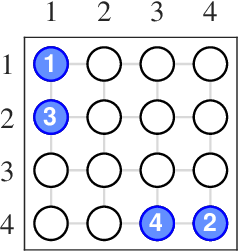} 
        \label{fig:30a-Omega_tilde}
    }
    \hfill
    \subfloat[]{
        \includegraphics[width=0.10\linewidth]{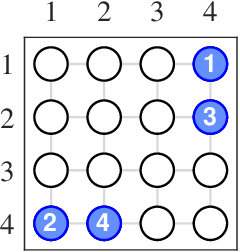}
        \label{fig:30b-Omega_tilde}
    }
    \hfill
    \subfloat[]{
        \includegraphics[width=0.10\linewidth]{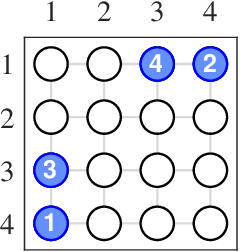}
        \label{fig:30c-Omega_tilde}
    }
    \hfill
    \subfloat[]{
        \includegraphics[width=0.10\linewidth]{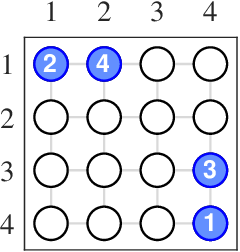} 
        \label{fig:30d-Omega_tilde}
    }
    \hfill
    \subfloat[]{
        \includegraphics[width=0.10\linewidth]{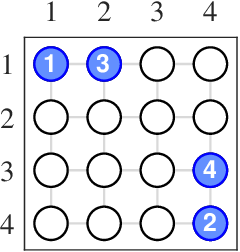} 
        \label{fig:30e-Omega_tilde}
    }
    \hfill
    \subfloat[]{
        \includegraphics[width=0.10\linewidth]{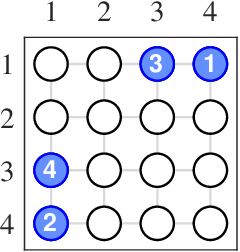}
        \label{fig:30f-Omega_tilde}
    }
    \hfill
    \subfloat[]{
        \includegraphics[width=0.10\linewidth]{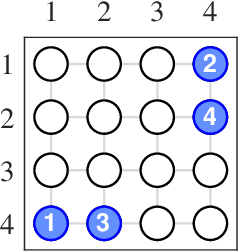}
        \label{fig:30g-Omega_tilde}
    }
    \hfill
    \subfloat[]{
        \includegraphics[width=0.10\linewidth]{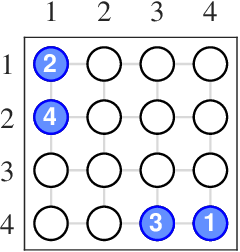} 
        \label{fig:30h-Omega_tilde}
    } 
    \caption{The antenna selection and permutation strategies $\tilde{\boldsymbol{\Omega}}$ corresponding to the AOD estimations in Fig. \ref{fig:noiseless-case}: (a)--(h) constitute the basic strategy group corresponding to the $8$ AOD solutions in Fig. \ref{fig:noiseless-case}(d) as well as those on the outer dashed curve in Fig. \ref{fig:noiseless-case}(e); (i)--(p) form the enhanced strategy group corresponding to the AOD solutions on the inner dashed curve in Fig.~\ref{fig:noiseless-case}(e). Circled numbers $1,\dots, 4$ represent the permutation order of selected antennas.}
    \label{fig:Omega_tilde}
\end{figure*}

In the noiseless case, we generate $10^5$ locations of Alice relative to Bob and Eve, and check the AOD estimated by Bob and that estimated by Eve.  In Fig. \ref{fig:noiseless-case}, we show several representative examples due to the space limitation. Under each realization, Bob performs exhaustive search over $\theta_{\mathrm{Bob}}$ and $\phi_{\mathrm{Bob}}$ given its received signal and the applied $\boldsymbol{\Omega}$, while Eve performs exhaustive search of $\theta_{\mathrm{Eve}}$, $\phi_{\mathrm{Eve}}$ and $\boldsymbol{\Omega}$ based on its received signal. It is observed from Fig. \ref{fig:noiseless-case} (and also many other numerical examples) that with known $\boldsymbol{\Omega}$, Bob is always able to uniquely estimate Alice's AOD. However, Eve may see multiple possible solutions of Alice's antenna strategy and AOD that lead to its received signal due to the lack of knowledge about $\boldsymbol{\Omega}$. 
Moreover, Fig. \ref{fig:noiseless-case} validates the estimation ambiguity characterized in Theorem \ref{thm:eight-ambiguity}. Specifically, Fig. \ref{fig:noiseless-case}\subref{fig:u0-success-2-solutions} shows two solutions given by (\ref{eq:aod_4u_1}) and (\ref{eq:aod_4u_2}). Figs. \ref{fig:noiseless-case}\subref{fig:4-AOD-solutions} and \ref{fig:noiseless-case}\subref{fig:4-AOD-solutions-2} show four solutions from (\ref{eq:aod_8_1}) to (\ref{eq:aod_8_4}) because the applied strategy violates condition c) in Theorem \ref{thm:eight-ambiguity} for a $4\times 2$ UPA. By contrast, when considering a $4\times 4$ UPA in Figs. \ref{fig:noiseless-case}\subref{fig:8-AOD-solutions} and \ref{fig:noiseless-case}\subref{fig:16-AOD-solutions} where condition c) holds, Eve observes all the AOD solutions in (\ref{eq:aod_8_1})--(\ref{eq:aod_8_8}). Note that in Fig. \ref{fig:noiseless-case}\subref{fig:16-AOD-solutions}, the AOD solutions given by \eqref{eq:aod_8_1}--\eqref{eq:aod_8_8} are those on the outer dashed curve, while those on the inner dashed curve form an enhanced AOD solution set introduced in Section \ref{subsec:enhanced-ambiguity}. 
Therefore, the comparison between Figs. \ref{fig:noiseless-case}\subref{fig:4-AOD-solutions-2} and \ref{fig:noiseless-case}\subref{fig:16-AOD-solutions} also demonstrates that expanding the UPA size along both dimensions (e.g., from a $4 \times 2$ to a $4 \times 4$ UPA) makes it possible, at least for specific true AODs, to induce additional AOD solutions for Eve beyond the basic solution set.

In Fig. \ref{fig:Omega_tilde}, we illustrate the corresponding strategies $\tilde{\boldsymbol{\Omega}}$ estimated by Eve that lead to the AOD estimations in Fig. \ref{fig:noiseless-case}. 
As observed from the results, the strategies within each group are essentially related to each other via spatial reflections and rotations, thereby generating the symmetric and structured AOD solutions in Fig.~\ref{fig:noiseless-case}. Meanwhile, it is straightforward to see that only the four strategies in Figs. \ref{fig:Omega_tilde}\subref{fig:30a-Omega}--\ref{fig:Omega_tilde}\subref{fig:30d-Omega} can physically fit into a $4 \times 2$ UPA. Consequently, when considering a $4 \times 2$ UPA instead of a $4 \times 4$ UPA, the $8$ and $16$ AOD solutions shown in Figs. \ref{fig:noiseless-case}\subref{fig:8-AOD-solutions} and \ref{fig:noiseless-case}\subref{fig:16-AOD-solutions} inevitably collapse to the four solutions in Figs. \ref{fig:noiseless-case}\subref{fig:4-AOD-solutions} and \ref{fig:noiseless-case}\subref{fig:4-AOD-solutions-2}, respectively.

Next, we consider the noisy case. The SNR of Bob is fixed to be $20 \ \mathrm{dB}$, and that of Eve varies from $10 \ \mathrm{dB}$ to $30 \ \mathrm{dB}$. We generate $10^5$ realizations of Alice's AOD relative to Bob and Eve. Under each realization, we utilize the maximum likelihood technique to estimate $\theta_{\mathrm{Bob}}$, $\phi_{\mathrm{Bob}}$ with knowledge of $\boldsymbol{\Omega}$ and $\theta_{\mathrm{Eve}}$, $\phi_{\mathrm{Eve}}$ without knowledge of $\boldsymbol{\Omega}$. 

Fig. \ref{fig:basic-accuracy-probability-vs-SNR} shows the AOD estimation accuracy for Bob and Eve versus Eve's SNR, where an estimated AOD is claimed to be accurate if its difference to the true AOD is within $5$ degrees. Here we consider a $4 \times 2$ UPA. It is observed that the applied strategy $\boldsymbol{\Omega} = ((1,1), (4,1), (1,2), (3,2))$ (and also other strategies in $\boldsymbol{\Lambda}$) can enable Bob to accurately estimate Alice's AOD, but prevent Eve from eavesdropping on Alice's AOD, in the noisy case.
\begin{figure}[t]
    \centering
        \includegraphics[width=0.74\linewidth]{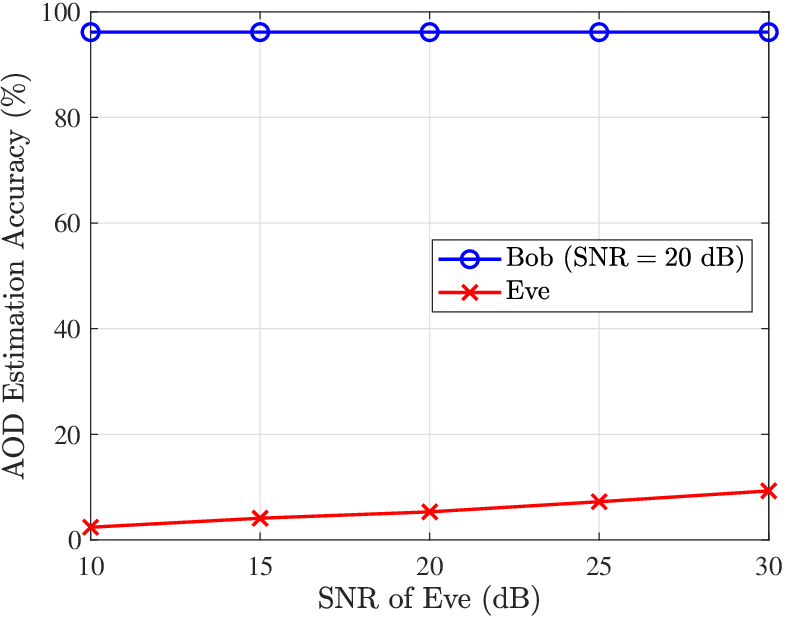}
    \caption{AOD estimation accuracy versus SNR of Eve.} 
    \label{fig:basic-accuracy-probability-vs-SNR}
\end{figure}
\begin{figure}[t]
    \centering
        \includegraphics[width=0.74\linewidth]{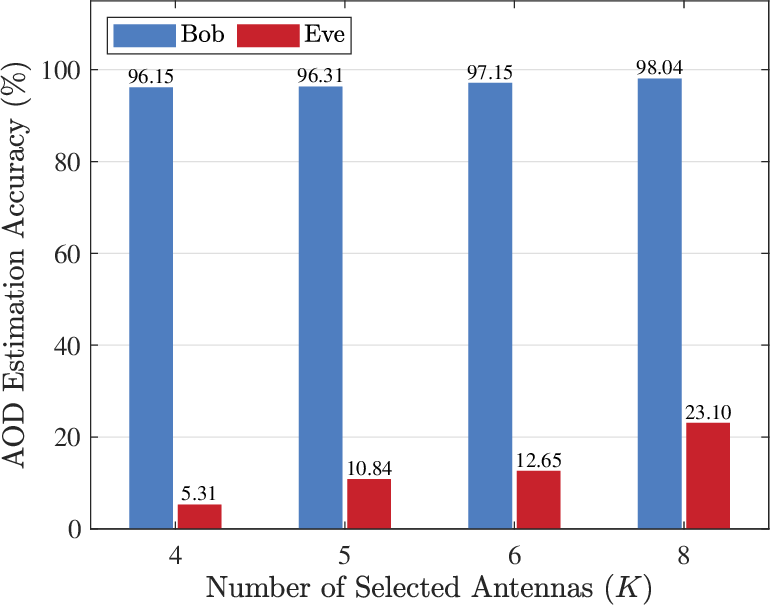}
    \caption{AOD estimation accuracy versus the number of selected antennas $K$.} 
    \label{fig:basic-accuracy-probability-vs-K}
\end{figure}

Fig. \ref{fig:basic-accuracy-probability-vs-K} presents the AOD estimation accuracy for Bob and Eve versus the number of selected antennas $K$ in a $4 \times 2$ UPA, where the SNRs of both Bob and Eve are fixed to be $20 \ \mathrm{dB}$. Specifically, for $K=4$, the applied strategy is the same as that in Fig. \ref{fig:basic-accuracy-probability-vs-SNR}, and antennas $(2,2)$, $(3,1)$, and $(2,1),(4,2)$ are sequentially added to realize the strategies for $K=5$, $6$, and $8$, respectively. As observed from the results, Bob maintains a consistently high estimation accuracy for all $K$, while Eve's AOD estimation accuracy increases significantly as $K$ increases. It is worth noting that the case of $K=8$ means Alice does not perform antenna selection. 
In this case, Eve exhibits an accuracy of $23.10\%$, which closely approaches the expected accuracy of $25\%$, since selecting all eight antennas from a $4 \times 2$ UPA generally leaves only the four AOD solutions given by (\ref{eq:aod_8_1}) to (\ref{eq:aod_8_4}) in Theorem \ref{thm:eight-ambiguity}, among which Eve can only make a random selection without additional prior information. This estimation accuracy is greatly higher than that with smaller $K$, which demonstrates the significance of transmit antenna selection in preserving localization privacy.

\subsection{Enhanced Privacy}
\begin{figure}[t]
    \centering
        \includegraphics[width=0.74\linewidth]{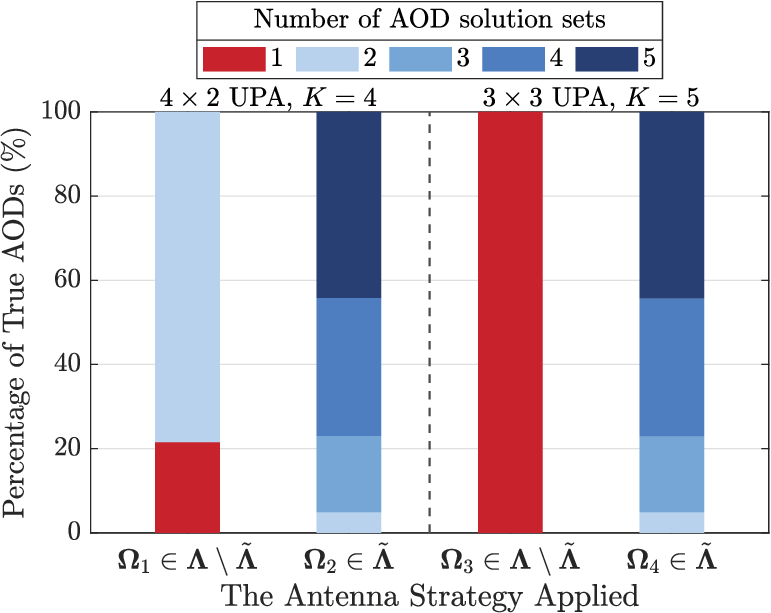}
    \caption{Percentage of true AODs corresponding to different numbers of AOD solution sets at Eve for different strategies.} 
    \label{fig:percentage-true-AODs-for-different-strategies}
\end{figure}
In this subsection, we validate the effectiveness of the scheme proposed in Section \ref{subsec:enhanced-ambiguity} to achieve enhanced privacy. First, in the noiseless case, we consider the following two setups: (a) $4 \times 2$ UPA, $K=4$, and (b) $3 \times 3$ UPA, $K=5$. Under each setup, we can obtain the set $\boldsymbol{\Lambda}$, and also construct the subset $\tilde{\boldsymbol{\Lambda}} \subset \boldsymbol{\Lambda}$ as in Theorem \ref{thm:Q-level-construction} given $Q=2$. Due to the space limitation, we select one pair of strategies from each setup for comparison: (a) $\boldsymbol{\Omega}_1 = ((1,1),(4,1),(1,2),(3,2)) \in \boldsymbol{\Lambda} \setminus \tilde{\boldsymbol{\Lambda}}$, $\boldsymbol{\Omega}_2 = ((1,1),(2,1),(2,2),(3,2)) \in \tilde{\boldsymbol{\Lambda}}$, and (b) $\boldsymbol{\Omega}_3 = ((1,1), (2,1), (2,2), (2,3), (1,3)) \in \boldsymbol{\Lambda} \setminus \tilde{\boldsymbol{\Lambda}}$, $\boldsymbol{\Omega}_4 = ((1,1), (2,1), (2,2), (2,3), (3,3))  \in \tilde{\boldsymbol{\Lambda}}$. We generate $10^5$ realizations of Alice's AOD relative to Eve, and check the number of AOD solution sets estimated by Eve under each realization. Note that Eve performs exhaustive search of $\boldsymbol{\Omega}$ over the set $\boldsymbol{\Lambda}$ when $\boldsymbol{\Omega}_1$ or $\boldsymbol{\Omega}_3$ is applied, while searching over the subset $\tilde{\boldsymbol{\Lambda}} \subset \boldsymbol{\Lambda}$ for $\boldsymbol{\Omega}_2$ and $\boldsymbol{\Omega}_4$.

Fig. \ref{fig:percentage-true-AODs-for-different-strategies} plots the percentage of the true AODs corresponding to different numbers of AOD solution sets for these applied strategies. 
It is observed that applying $\boldsymbol{\Omega}_1$ and $\boldsymbol{\Omega}_3$ results in cases with only a single AOD solution set at Eve. In sharp contrast, $\boldsymbol{\Omega}_2$ and $\boldsymbol{\Omega}_4$ completely eliminate the single-solution-set case, guaranteeing that Eve is always confused by at least two solution sets for any true AOD. Furthermore, $\boldsymbol{\Omega}_2$ and $\boldsymbol{\Omega}_4$ exhibit remarkably similar percentage distributions across different numbers of solution sets, which stems from our structured construction in Theorem~\ref{thm:Q-level-construction}. This also demonstrates that our proposed enhanced privacy design is applicable to diverse UPA sizes and numbers of selected antennas. 

In Fig. \ref{fig:AOD-solution-sets}, we present two examples of AOD estimation by Bob and Eve under $\boldsymbol{\Omega}_4$. It is observed from Fig. \ref{fig:AOD-solution-sets}\subref{fig:32-AOD-solutions} and 
Fig. \ref{fig:AOD-solution-sets}\subref{fig:40-AOD-solutions} that Eve can estimate four and five AOD solution sets, respectively, which indicates that the applied strategy from $\tilde{\boldsymbol{\Lambda}}$ achieves better privacy in the noiseless case.
\begin{figure}[t]
    \centering
    \subfloat[]{
        \includegraphics[trim=0cm 0cm 0.8cm 1.0cm, clip, width=0.47\linewidth]{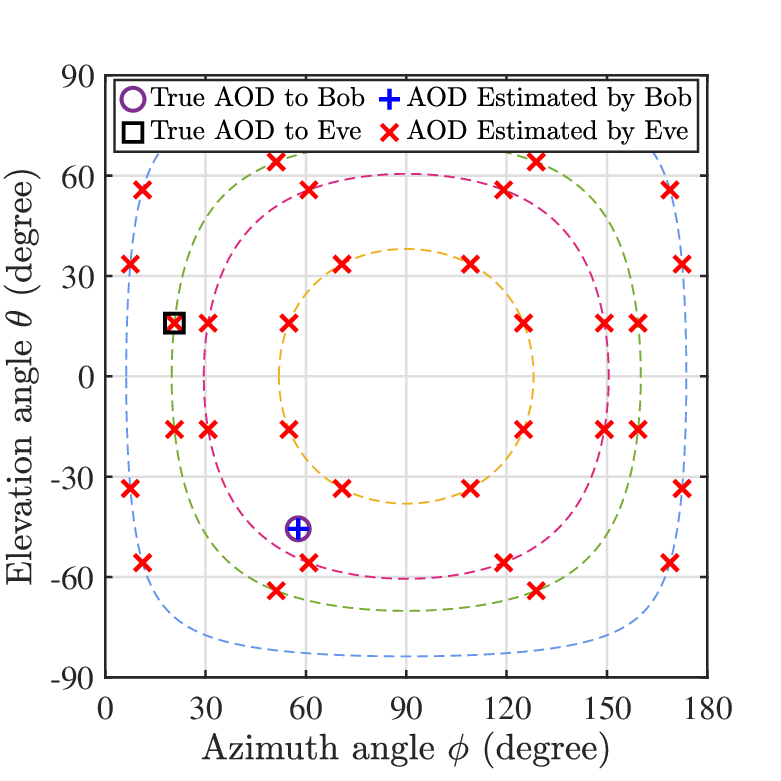}
        \label{fig:32-AOD-solutions}
    }
    \hfill
    \subfloat[]{
        \includegraphics[trim=0cm 0cm 0.8cm 1.0cm, clip, width=0.47\linewidth]{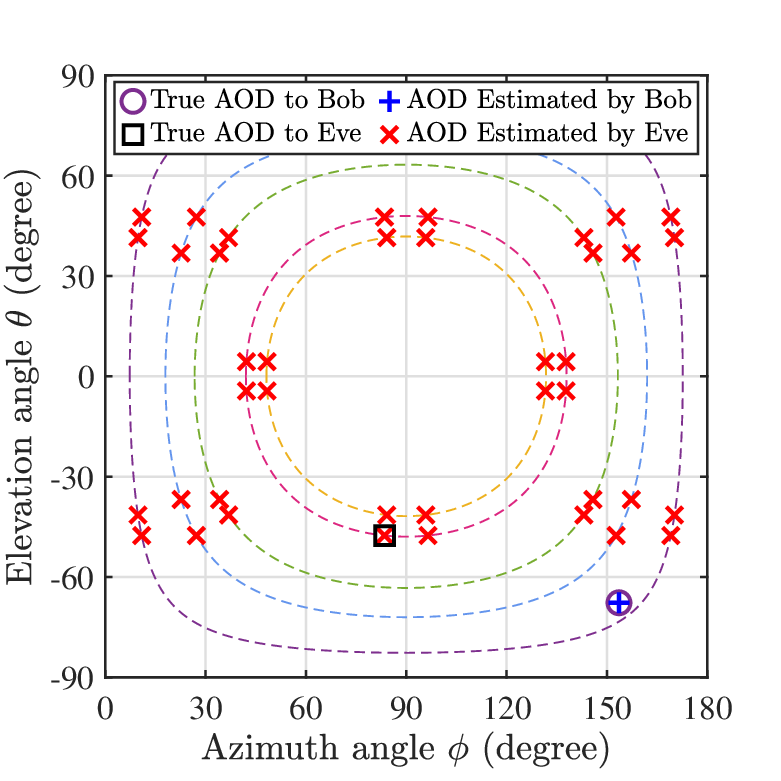}
        \label{fig:40-AOD-solutions}
    }
    \caption{AOD estimation by Bob and Eve in the noiseless case under $\boldsymbol{\Omega}_4$: 
    (a) $\theta_{\mathrm{Bob}}=-45.6^\circ$, $\phi_{\mathrm{Bob}}=57.7^\circ$, $\theta_{\mathrm{Eve}}\ = 15.9^\circ$, $\phi_{\mathrm{Eve}} = 20.7^\circ$; (b) $\theta_{\mathrm{Bob}}= -67.7^\circ$, $\phi_{\mathrm{Bob}}=153.6^\circ$, $\theta_{\mathrm{Eve}} = -47.6^\circ$, $\phi_{\mathrm{Eve}} = 83.5^\circ$.
    }
    \label{fig:AOD-solution-sets}
\end{figure}

Next, we consider the noisy case. In Fig.~\ref{fig:enhanced-accuracy-probability-vs-SNR}, we compare the strategies $\boldsymbol{\Omega}_3$ and $\boldsymbol{\Omega}_4$ as considered in Fig. \ref{fig:percentage-true-AODs-for-different-strategies} and evaluate Eve's AOD estimation accuracy versus its SNR. It is observed that $\boldsymbol{\Omega}_4$ consistently restricts Eve to a lower AOD estimation accuracy than $\boldsymbol{\Omega}_3$ over the entire SNR range. More importantly, Eve's accuracy under $\boldsymbol{\Omega}_4$ gradually saturates as the SNR increases. This indicates that although increasing the SNR mitigates the impact of noise, it cannot eliminate the intrinsic estimation ambiguity induced by multiple AOD solution sets that exist in the noiseless case. Therefore, our proposed enhanced privacy design can maintain its privacy advantage even in the high-SNR regime.
\begin{figure}[t]
    \centering
        \includegraphics[width=0.74\linewidth]{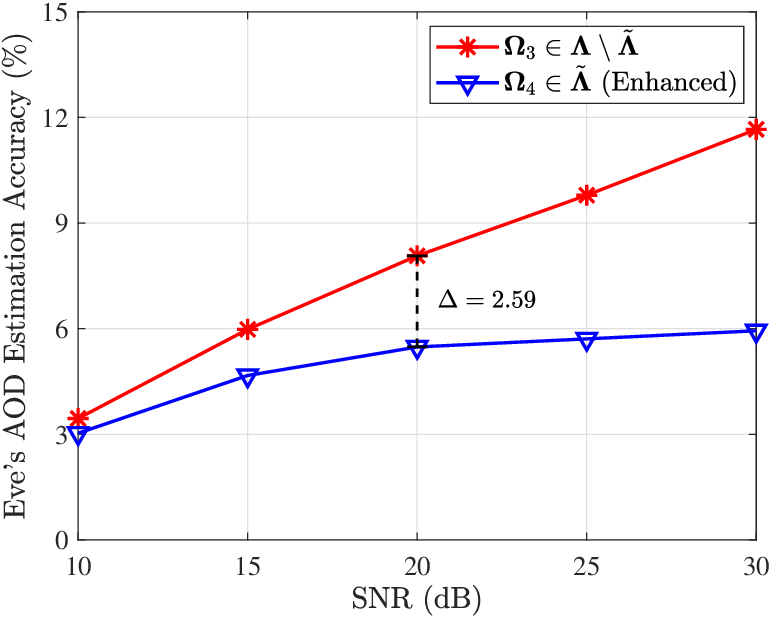}
    \caption{Eve's AOD estimation accuracy versus SNR under $\boldsymbol{\Omega}_3$ and $\boldsymbol{\Omega}_4$.}
    \label{fig:enhanced-accuracy-probability-vs-SNR}
\end{figure}

\section{Conclusions} \label{sec:conclusions}

This paper proposed a transmit antenna selection and permutation scheme to preserve localization privacy. In analogy to cryptography in secure communication, the transmit antenna selection and permutation strategy serves as the secret key shared between the legitimate transmitter and receiver to prevent Eve from estimating Alice's AOD. 
To guarantee unique AOD estimation in the noiseless case for legitimate localization, we characterized the set of all the feasible antenna selection and permutation strategies that effectively eliminate grating lobes. Furthermore, we proved that any antenna selection and permutation strategy in the above set inherently induces AOD estimation ambiguity to Eve in the noiseless case, and revealed that the estimated AODs at Eve occur in groups. Therefore, we constructed a subset of strategies to achieve enhanced privacy by making Eve always confused by multiple groups of AODs. Numerical results demonstrated the effectiveness of our proposed scheme in the noisy case as well.
We hope that this work can motivate more efforts to investigate sensing privacy issues in 6G ISAC.

% \newpage
% \ 

% \appendix
\begin{appendices}

\section{Proof of Proposition \ref{prop:uniqueness}} \label{app:proof-theorem-Bob}
    Note that when $G\geq K$, (\ref{equ:requirement-def1}) is equivalent to
    \begin{equation}
        \tilde{\boldsymbol{h}}(\boldsymbol{\Omega},\tilde{\theta}_{\mathrm{Bob}}, \tilde{\phi}_{\mathrm{Bob}}) \neq \tilde{\boldsymbol{h}}(\boldsymbol{\Omega}, {\theta}_{\mathrm{Bob}}, {\phi}_{\mathrm{Bob}}). \label{equ:Bob-ambiguity-receive-equal}
    \end{equation}
    In the following, we prove the necessary and sufficient parts of Proposition \ref{prop:uniqueness}.
    
    \textit{1) Proof of Necessity: }
  
    Suppose there exist $u$ and $v$ with $0 < u^2 + v^2 < 1$ such that 
    \(u \boldsymbol{x}(\boldsymbol{\Omega}) + v\boldsymbol{z}(\boldsymbol{\Omega}) = \boldsymbol{k} \in \mathbb{Z}^{K \times 1}\). Given $\tilde{u}_{\mathrm{Bob}} = u$, $\tilde{v}_{\mathrm{Bob}} = v$, $u_{\mathrm{Bob}} = -u$, and $v_{\mathrm{Bob}} = -v$, it can be easily verified that $(\tilde{u}_{\mathrm{Bob}}, \tilde{v}_{\mathrm{Bob}}) \neq ({u}_{\mathrm{Bob}}, {v}_{\mathrm{Bob}})$ and $0< \tilde{u}_{\mathrm{Bob}}^2 + \tilde{v}_{\mathrm{Bob}}^2 = {u}_{\mathrm{Bob}}^2 + {v}_{\mathrm{Bob}}^2 <1$. Moreover, we observe that
    \begin{equation}
        \tilde{u}_{\mathrm{Bob}} \boldsymbol{x}(\boldsymbol{\Omega})+ \tilde{v}_{\mathrm{Bob}} \boldsymbol{z}(\boldsymbol{\Omega}) =  {u}_{\mathrm{Bob}} \boldsymbol{x}(\boldsymbol{\Omega})+  {v}_{\mathrm{Bob}} \boldsymbol{z}(\boldsymbol{\Omega})+ 2 \boldsymbol{k}.  \label{equ:proof-necessity}
    \end{equation}
    Using the complex exponential function and because ${\gamma_{\mathrm{Bob}}}$ is unknown, it follows from (\ref{equ:proof-necessity}) that $\tilde{\boldsymbol{h}}(\boldsymbol{\Omega},\tilde{\theta}_{\mathrm{Bob}}, \tilde{\phi}_{\mathrm{Bob}}) = \tilde{\boldsymbol{h}}(\boldsymbol{\Omega}, {\theta}_{\mathrm{Bob}}, {\phi}_{\mathrm{Bob}})$, which is a contradiction to (\ref{equ:Bob-ambiguity-receive-equal}).

    \textit{2) Proof of Sufficiency: }
    
   Suppose (\ref{equ:Bob-ambiguity-receive-equal}) does not hold. This implies that there exists \(\boldsymbol{k} \in \mathbb{Z}^{K \times 1}\) such that
    $
    \tilde{u}_{\mathrm{Bob}} \boldsymbol{x}(\boldsymbol{\Omega})+ \tilde{v}_{\mathrm{Bob}} \boldsymbol{z}(\boldsymbol{\Omega}) =  {u}_{\mathrm{Bob}} \boldsymbol{x}(\boldsymbol{\Omega})+  {v}_{\mathrm{Bob}} \boldsymbol{z}(\boldsymbol{\Omega})+ 2 \boldsymbol{k},
    $
    due to the periodicity of the complex exponential function in (\ref{equ:Bob-tilde-h-B-u-v}).
    Let \(\Delta u = \tilde{u}_{\mathrm{Bob}} - {u}_{\mathrm{Bob}}\), and \(\Delta v  = \tilde{v}_{\mathrm{Bob}} - {v}_{\mathrm{Bob}}\), then we have \(  \boldsymbol{k} =  (\Delta u/2 )\boldsymbol{x}(\boldsymbol{\Omega})+  (\Delta v / 2)\boldsymbol{z}(\boldsymbol{\Omega})\).
    It is easy to verify that
    \begin{align}
    \nonumber & \left( \Delta u/2 \right)^2+\left(\Delta v/2\right)^2
    \stackrel{(a)}{<} \frac{1}{2}\left(1 - \tilde{u}_{\mathrm{Bob}} {u}_{\mathrm{Bob}} - \tilde{v}_{\mathrm{Bob}} {v}_{\mathrm{Bob}}\right) \\
    % \nonumber & =\frac{1}{2}\Big(1 - \cos \tilde{\theta}_{\mathrm{Bob}} \cos {\theta}_{\mathrm{Bob}} \cos \tilde{\phi}_{\mathrm{Bob}} \cos {\phi}_{\mathrm{Bob}} 
    %  - \sin \tilde{\theta}_{\mathrm{Bob}} \sin {\theta}_{\mathrm{Bob}}\Big) \\
    &\stackrel{(b)}{<} \frac{1}{2}(1 + \cos (\tilde{\theta}_{\mathrm{Bob}} + {\theta}_{\mathrm{Bob}}))  \stackrel{(c)}{\leq} 1,   \label{equ:sum-delta-uv/2^2} 
    \end{align}
    where $(a)$ is due to $\tilde{u}_{\mathrm{Bob}}^2+\tilde{v}_{\mathrm{Bob}}^2 < 1$ and $u_{\mathrm{Bob}}^2+v_{\mathrm{Bob}}^2 < 1$,  $(b)$ is due to (\ref{equ:u-cos-cos-v-sin}), the trigonometric identities, and the facts that $\cos \tilde{\phi}_{\mathrm{Bob}}, \cos {\phi}_{\mathrm{Bob}} \in (-1, 1)$, and $\cos \tilde{\theta}_{\mathrm{Bob}},$ $\cos {\theta}_{\mathrm{Bob}} \in (0, 1]$, $(c)$ follows from $\cos (\tilde{\theta}_{\mathrm{Bob}} + {\theta}_{\mathrm{Bob}}) \in [-1, 1]$. From (\ref{equ:sum-delta-uv/2^2}) and because $\Delta u$ and $\Delta v$ are not both zero, a solution of $u$ and $v$ with \(0< u^2 + v^2<1\) such that \(u \boldsymbol{x}(\boldsymbol{\Omega}) + v \boldsymbol{z}(\boldsymbol{\Omega}) \in \mathbb{Z}^{K \times 1}\) can be readily found, i.e., \(u = \Delta u/2\) and 
    \(v = \Delta v /2\). This contradicts the characterization of $\boldsymbol{\Lambda}$ in (\ref{equ:theorem1}).

    Proposition \ref{prop:uniqueness} is proved.
    
\section{Proof of Theorem \ref{thm:equ-thm-uniqueness-K-equal-3}} \label{app:equ-thm-uniqueness-K-equal-3}
   
    Given any $\boldsymbol{\Omega} \in \boldsymbol{\Lambda}$, define the open elliptic region $\mathcal{E}(\boldsymbol{\Omega})=\{\hat{\boldsymbol{P}}(\boldsymbol{\Omega})\boldsymbol{w} \mid \boldsymbol{w} \in \mathcal{D}\}$, where $\hat{\boldsymbol{P}}(\boldsymbol{\Omega}) = [\boldsymbol{P}(\boldsymbol{\Omega})]_{2:3, :}$ and $\mathcal{D} = \{\boldsymbol{w} \in \mathbb{R}^2 \mid \|\boldsymbol{w}\|_2 < 1\}$ denotes the open unit disk. Note that if $\operatorname{rank}(\hat{\boldsymbol{P}}(\boldsymbol{\Omega})) < 2$, we can always construct some $\boldsymbol{w} \in \mathcal{D}$ such that $\hat{\boldsymbol{P}}(\boldsymbol{\Omega})\boldsymbol{w} = \mathbf{0}$, which contradicts Proposition \ref{prop:uniqueness}. Therefore, $\operatorname{rank}(\hat{\boldsymbol{P}}(\boldsymbol{\Omega})) = 2$ must hold, and the region $\mathcal{E}(\boldsymbol{\Omega})$ will not degenerate into a line. Next, it is not difficult to see that $\hat{\boldsymbol{P}}(\boldsymbol{\Omega})$ maps $\mathcal{D} \setminus \{\mathbf{0}\}$ bijectively onto $\mathcal{E}(\boldsymbol{\Omega}) \setminus \{\mathbf{0}\}$. By Proposition \ref{prop:uniqueness}, it thus suffices to prove that $\mathcal{E}(\boldsymbol{\Omega})$ contains no integer point other than the origin if and only if $|\det(\hat{\boldsymbol{P}}(\boldsymbol{\Omega}))| = 1$. We prove the necessary and sufficient parts as follows:

    \textit{1) Proof of Necessity:}

    Assume that $\mathcal{E}(\boldsymbol{\Omega})$ contains no integer point other than the origin. By Minkowski's theorem, the area of the region $\mathcal{E}(\boldsymbol{\Omega})$ must be no greater than four, i.e., $\pi |\operatorname{det}(\hat{\boldsymbol{P}}(\boldsymbol{\Omega}))| \leq 4$. Since $\operatorname{det}(\hat{\boldsymbol{P}}(\boldsymbol{\Omega}))$ is an integer and $\operatorname{det}(\hat{\boldsymbol{P}}(\boldsymbol{\Omega})) \neq 0$, it follows that $|\operatorname{det}(\hat{\boldsymbol{P}}(\boldsymbol{\Omega}))| = 1$.

    \textit{2) Proof of Sufficiency:}

    Assume that $|\det(\hat{\boldsymbol{P}}(\boldsymbol{\Omega}))| = 1$, and suppose, for contradiction, there exists a point $\boldsymbol{e} \in \left(\mathbb{Z}^2 \setminus\{\mathbf{0}\}\right) \cap \mathcal{E}(\boldsymbol{\Omega})$. From the definition of $\mathcal{E}(\boldsymbol{\Omega})$, it follows that $(\hat{\boldsymbol{P}}(\boldsymbol{\Omega}))^{-1} \boldsymbol{e} \in \mathcal{D}$, where $(\hat{\boldsymbol{P}}(\boldsymbol{\Omega}))^{-1}$ denotes the inverse of $\hat{\boldsymbol{P}}(\boldsymbol{\Omega})$. Since $\hat{\boldsymbol{P}}(\boldsymbol{\Omega})$ is an integer matrix with $|\operatorname{det}(\hat{\boldsymbol{P}}(\boldsymbol{\Omega}))| = 1$, $(\hat{\boldsymbol{P}}(\boldsymbol{\Omega}))^{-1}$ must also be an integer matrix. This implies $(\hat{\boldsymbol{P}}(\boldsymbol{\Omega}))^{-1} \boldsymbol{e} \in (\mathbb{Z}^2 \setminus\{\mathbf{0}\}) \cap \mathcal{D}$. Due to the fact that $\mathcal{D}$ contains no integer point other than the origin, we have found a contradiction, implying that no such $\boldsymbol{e}$ exists.

    Theorem \ref{thm:equ-thm-uniqueness-K-equal-3} is proved.

    \section{Proof of Proposition \ref{prop:Eve-ambiguity}} \label{app:proof-theorem-Eve}
    Note that when $G\geq K$, (\ref{equ:requirement-def2}) is equivalent to
    \begin{equation}
        \tilde{\boldsymbol{h}}_{\mathrm{Eve}}(\tilde{\boldsymbol{\Omega}}, \tilde{\theta}_{\mathrm{Eve}}, \tilde{\phi}_{\mathrm{Eve}}) = \tilde{\boldsymbol{h}}_{\mathrm{Eve}}(\boldsymbol{\Omega}, {\theta}_{\mathrm{Eve}}, {\phi}_{\mathrm{Eve}}). \label{equ:Eve-ambiguity-receive-equal}
    \end{equation}
    In the following, we prove the necessary and sufficient parts of Proposition \ref{prop:Eve-ambiguity}.
    
    \textit{1) Proof of Necessity: }
    
    Suppose (\ref{equ:Eve-ambiguity-receive-equal}) holds. Then, $e^{-j \pi\left((\tilde{x}_k - \tilde{x}_1) \tilde{u}_{\mathrm{Eve}} + (\tilde{z}_k - \tilde{z}_1) \tilde{v}_{\mathrm{Eve}}\right)} = e^{-j \pi\left(({x}_k - {x}_1) {u}_{\mathrm{Eve}} + ({z}_k - {z}_1) {v}_{\mathrm{Eve}}\right)}, \forall k$, according to (\ref{equ:Eve-tilde-h-B-u-v}). Due to the periodicity of the complex exponential function, it can be then verified that $(\tilde{x}_k - \tilde{x}_1) \tilde{u}_{\mathrm{Eve}} + (\tilde{z}_k - \tilde{z}_1) \tilde{v}_{\mathrm{Eve}} - ({x}_k - {x}_1) {u}_{\mathrm{Eve}} - ({z}_k - {z}_1) {v}_{\mathrm{Eve}}$ is an even integer for all $k$.

    \textit{2) Proof of Sufficiency: }
    
    Given $u_{\mathrm{Eve}}$, $v_{\mathrm{Eve}}$, $\boldsymbol{x}(\boldsymbol{\Omega})$, and $\boldsymbol{z}(\boldsymbol{\Omega})$, suppose there exists an alternative solution of $\tilde{u}_{\mathrm{Eve}}$, $\tilde{v}_{\mathrm{Eve}}$, $\tilde{\boldsymbol{x}}(\tilde{\boldsymbol{\Omega}})$, and $\tilde{\boldsymbol{z}}(\tilde{\boldsymbol{\Omega}})$ such that
    \begin{equation}
    \tilde{u}_{\mathrm{Eve}} \tilde{\boldsymbol{x}}(\tilde{\boldsymbol{\Omega}})  +  \tilde{v}_{\mathrm{Eve}} \tilde{\boldsymbol{z}}(\tilde{\boldsymbol{\Omega}})= {u}_{\mathrm{Eve}} {\boldsymbol{x}}({\boldsymbol{\Omega}}) + {v}_{\mathrm{Eve}} {\boldsymbol{z}}({\boldsymbol{\Omega}}) +2 \boldsymbol{k}, \label{equ:proof-appB}
    \end{equation}
    where $\boldsymbol{k} \in \mathbb{Z}^{K \times 1}$.
    Using the complex exponential function, (\ref{equ:proof-appB}) can then be written as 
    $
        e^{-j \pi\left((\tilde{x}_k - \tilde{x}_1) \tilde{u}_{\mathrm{Eve}} + (\tilde{z}_k - \tilde{z}_1) \tilde{v}_{\mathrm{Eve}}\right)}  = e^{-j \pi\left(({x}_k - {x}_1) {u}_{\mathrm{Eve}} + ({z}_k - {z}_1) {v}_{\mathrm{Eve}}\right)}, \forall k.  \label{equ:proof2-exp-equal}
    $
    Note that ${\gamma_{\mathrm{Eve}}}$ is unknown. It can be observed that there exists a feasible solution of ${\tilde{\gamma}_{\mathrm{Eve}}}$ to satisfy 
    $
    {\tilde{\gamma}_{\mathrm{Eve}}}e^{-j\pi \left((\tilde{x}_1-1) \tilde{u}_{\mathrm{Eve}} + (\tilde{z}_1 -1) \tilde{v}_{\mathrm{Eve}}\right)}= {\gamma}_{\mathrm{Eve}} e^{-j\pi \left((x_1-1) u_{\mathrm{Eve}} + (z_1-1) v_{\mathrm{Eve}} \right)}. \label{equ:proof2-gamma-equal}
    $
    In this case, we can conclude that (\ref{equ:Eve-ambiguity-receive-equal}) holds according to (\ref{equ:Eve-tilde-h-B-u-v}). 

    Proposition \ref{prop:Eve-ambiguity} is proved.

\section{Proof of Theorem \ref{thm:eight-ambiguity}} \label{app:eight-ambiguity}

    We first suppose conditions a), b), and c) on $u_{\mathrm{Eve}}$, $v_{\mathrm{Eve}}$, and $\boldsymbol{\Omega}$ in Theorem \ref{thm:eight-ambiguity} are satisfied. Given any $\boldsymbol{\Omega} \in \boldsymbol{\Lambda}$, $u_{\mathrm{Eve}}$, and $v_{\mathrm{Eve}}$ with $0 < u_{\mathrm{Eve}}^2+v_{\mathrm{Eve}}^2 < 1$, multiple solutions of $\tilde{u}_{\mathrm{Eve}}$, $\tilde{v}_{\mathrm{Eve}}$, and $\tilde{\boldsymbol{\Omega}}$ can be easily obtained to satisfy the condition in Proposition \ref{prop:Eve-ambiguity}.
    For example, there must exist a strategy $\tilde{\boldsymbol{\Omega}}$ such that $[\tilde{\boldsymbol{x}}(\tilde{\boldsymbol{\Omega}}), \tilde{\boldsymbol{z}}(\tilde{\boldsymbol{\Omega}})] = [-\boldsymbol{x}(\boldsymbol{\Omega}), \boldsymbol{z}(\boldsymbol{\Omega})]$. Specifically, it can be explicitly expressed as $\tilde{\boldsymbol{\Omega}} = ( (\tilde{x}_1, \tilde{z}_1), (\tilde{x}_1 - (x_2 - x_1), \tilde{z}_1 + (z_2 - z_1)), \dots, (\tilde{x}_1 - (x_K - x_1), \tilde{z}_1 + (z_K - z_1)))$ by selecting a valid reference antenna $(\tilde{x}_1, \tilde{z}_1)$ within the UPA. Based on Proposition \ref{prop:uniqueness} and the fact that ${\boldsymbol{\Omega}} \in \boldsymbol{\Lambda}$, it is not difficult to see $\tilde{\boldsymbol{\Omega}} \in \boldsymbol{\Lambda}$ as well. Given $\tilde{u}_{\mathrm{Eve}} = -u_{\mathrm{Eve}}$ and $\tilde{v}_{\mathrm{Eve}} = v_{\mathrm{Eve}}$, it can then be easily verified that $\tilde{u}_{\mathrm{Eve}}\tilde{\boldsymbol{x}}(\tilde{\boldsymbol{\Omega}}) + \tilde{v}_{\mathrm{Eve}}\tilde{\boldsymbol{z}}(\tilde{\boldsymbol{\Omega}}) - u_{\mathrm{Eve}}\boldsymbol{x}(\boldsymbol{\Omega}) - v_{\mathrm{Eve}}\boldsymbol{z}(\boldsymbol{\Omega}) = \boldsymbol{0}$. According to (\ref{equ:u-cos-cos-v-sin-Eve}), this implies that $(\tilde{\theta}_{\mathrm{Eve}}, \tilde{\phi}_{\mathrm{Eve}}) = (\theta_{\mathrm{Eve}}, \pi-\phi_{\mathrm{Eve}})$. 
    To summarize, the AOD estimations at Eve inherently include the following eight solutions:
    \begin{enumerate}
        \item $\tilde{\boldsymbol{\Omega}}$ such that $[\tilde{\boldsymbol{x}}(\tilde{\boldsymbol{\Omega}}), \tilde{\boldsymbol{z}}(\tilde{\boldsymbol{\Omega}})] = [\boldsymbol{x}(\boldsymbol{\Omega}), \boldsymbol{z}(\boldsymbol{\Omega})]$, and $(\tilde{u}_{\mathrm{Eve}},\tilde{v}_{\mathrm{Eve}}) = (u_{\mathrm{Eve}},v_{\mathrm{Eve}})$, which implies that $(\tilde{\theta}_{\mathrm{Eve}}, \tilde{\phi}_{\mathrm{Eve}}) = (\theta_{\mathrm{Eve}}, \phi_{\mathrm{Eve}})$;
        \item $\tilde{\boldsymbol{\Omega}}$ such that $[\tilde{\boldsymbol{x}}(\tilde{\boldsymbol{\Omega}}), \tilde{\boldsymbol{z}}(\tilde{\boldsymbol{\Omega}})] = [-\boldsymbol{x}(\boldsymbol{\Omega}), \boldsymbol{z}(\boldsymbol{\Omega})]$, and $(\tilde{u}_{\mathrm{Eve}},\tilde{v}_{\mathrm{Eve}}) = (-u_{\mathrm{Eve}},v_{\mathrm{Eve}})$, which implies that $(\tilde{\theta}_{\mathrm{Eve}}, \tilde{\phi}_{\mathrm{Eve}}) = (\theta_{\mathrm{Eve}}, \pi-\phi_{\mathrm{Eve}})$;
        \item $\tilde{\boldsymbol{\Omega}}$ such that $[\tilde{\boldsymbol{x}}(\tilde{\boldsymbol{\Omega}}), \tilde{\boldsymbol{z}}(\tilde{\boldsymbol{\Omega}})] = [\boldsymbol{x}(\boldsymbol{\Omega}), -\boldsymbol{z}(\boldsymbol{\Omega})]$, and $(\tilde{u}_{\mathrm{Eve}},\tilde{v}_{\mathrm{Eve}}) = (u_{\mathrm{Eve}},-v_{\mathrm{Eve}})$, which implies that $(\tilde{\theta}_{\mathrm{Eve}}, \tilde{\phi}_{\mathrm{Eve}}) = (-\theta_{\mathrm{Eve}}, \phi_{\mathrm{Eve}})$;
        \item $\tilde{\boldsymbol{\Omega}}$ such that $[\tilde{\boldsymbol{x}}(\tilde{\boldsymbol{\Omega}}), \tilde{\boldsymbol{z}}(\tilde{\boldsymbol{\Omega}})] = [-\boldsymbol{x}(\boldsymbol{\Omega}), -\boldsymbol{z}(\boldsymbol{\Omega})]$, and $(\tilde{u}_{\mathrm{Eve}},\tilde{v}_{\mathrm{Eve}}) = (-u_{\mathrm{Eve}},-v_{\mathrm{Eve}})$, which implies that $(\tilde{\theta}_{\mathrm{Eve}}, \tilde{\phi}_{\mathrm{Eve}}) = (-\theta_{\mathrm{Eve}}, \pi-\phi_{\mathrm{Eve}})$;
        \item $\tilde{\boldsymbol{\Omega}}$ such that $[\tilde{\boldsymbol{x}}(\tilde{\boldsymbol{\Omega}}), \tilde{\boldsymbol{z}}(\tilde{\boldsymbol{\Omega}})] = [\boldsymbol{z}(\boldsymbol{\Omega}), \boldsymbol{x}(\boldsymbol{\Omega})]$, and $(\tilde{u}_{\mathrm{Eve}},\tilde{v}_{\mathrm{Eve}}) = (v_{\mathrm{Eve}},u_{\mathrm{Eve}})$, which implies that $(\tilde{\theta}_{\mathrm{Eve}}, \tilde{\phi}_{\mathrm{Eve}}) = (\arcsin u_{\mathrm{Eve}}, \arccos\frac{v_{\mathrm{Eve}}}{\sqrt{1-u_{\mathrm{Eve}}^2}})$;
        \item $\tilde{\boldsymbol{\Omega}}$ such that $[\tilde{\boldsymbol{x}}(\tilde{\boldsymbol{\Omega}}), \tilde{\boldsymbol{z}}(\tilde{\boldsymbol{\Omega}})] = [-\boldsymbol{z}(\boldsymbol{\Omega}), \boldsymbol{x}(\boldsymbol{\Omega})]$, and $(\tilde{u}_{\mathrm{Eve}},\tilde{v}_{\mathrm{Eve}}) = (-v_{\mathrm{Eve}},u_{\mathrm{Eve}})$, which implies that $(\tilde{\theta}_{\mathrm{Eve}}, \tilde{\phi}_{\mathrm{Eve}}) = (\arcsin u_{\mathrm{Eve}}, \pi - \arccos\frac{v_{\mathrm{Eve}}}{\sqrt{1-u_{\mathrm{Eve}}^2}})$;
        \item $\tilde{\boldsymbol{\Omega}}$ such that $[\tilde{\boldsymbol{x}}(\tilde{\boldsymbol{\Omega}}), \tilde{\boldsymbol{z}}(\tilde{\boldsymbol{\Omega}})] = [\boldsymbol{z}(\boldsymbol{\Omega}), -\boldsymbol{x}(\boldsymbol{\Omega})]$, and $(\tilde{u}_{\mathrm{Eve}},\tilde{v}_{\mathrm{Eve}}) = (v_{\mathrm{Eve}},-u_{\mathrm{Eve}})$, which implies that $(\tilde{\theta}_{\mathrm{Eve}}, \tilde{\phi}_{\mathrm{Eve}}) = (-\arcsin u_{\mathrm{Eve}}, \arccos\frac{v_{\mathrm{Eve}}}{\sqrt{1-u_{\mathrm{Eve}}^2}})$;
        \item $\tilde{\boldsymbol{\Omega}}$ such that $[\tilde{\boldsymbol{x}}(\tilde{\boldsymbol{\Omega}}), \tilde{\boldsymbol{z}}(\tilde{\boldsymbol{\Omega}})] = [-\boldsymbol{z}(\boldsymbol{\Omega}), -\boldsymbol{x}(\boldsymbol{\Omega})]$, and $(\tilde{u}_{\mathrm{Eve}},\tilde{v}_{\mathrm{Eve}}) = (-v_{\mathrm{Eve}},-u_{\mathrm{Eve}})$, which implies that $(\tilde{\theta}_{\mathrm{Eve}}, \tilde{\phi}_{\mathrm{Eve}}) = (-\arcsin u_{\mathrm{Eve}}, \pi - \arccos\frac{v_{\mathrm{Eve}}}{\sqrt{1-u_{\mathrm{Eve}}^2}})$.
    \end{enumerate}
    In this case, the AOD estimations $(\tilde{\theta}_{\mathrm{Eve}}, \tilde{\phi}_{\mathrm{Eve}})$ listed above, referred to as AODs 1 to 8, are all distinct.
    
    However, if condition c) is satisfied, but condition a) or b) is violated, AODs 1 to 8 numerically collapse to four distinct solutions. Three sub-cases are as follows. If condition c) is satisfied and $u_{\mathrm{Eve}} = 0$, $v_{\mathrm{Eve}} \neq 0$, which implies that $\phi_{\mathrm{Eve}} = \pi/2$, then AOD 1 and AOD 2 collapse to $(\theta_{\mathrm{Eve}}, \pi/2)$, AOD 3 and AOD 4 collapse to $(-\theta_{\mathrm{Eve}}, \pi/2)$, AOD 5 and AOD 7 collapse to $(0, \arccos v_{\mathrm{Eve}})$, and AOD 6 and AOD 8 collapse to $(0, \pi - \arccos v_{\mathrm{Eve}})$. Similarly, if condition c) is satisfied and $u_{\mathrm{Eve}} \neq 0$, $v_{\mathrm{Eve}} = 0$, which implies that $\theta_{\mathrm{Eve}} = 0$, then AOD 1 and AOD 3 collapse to $(0, \phi_{\mathrm{Eve}})$, AOD 2 and AOD 4 collapse to $(0, \pi - \phi_{\mathrm{Eve}})$, AOD 5 and AOD 6 collapse to $(\arcsin u_{\mathrm{Eve}}, \pi/2)$, and AOD 7 and AOD 8 collapse to $(-\arcsin u_{\mathrm{Eve}}, \pi/2)$. Moreover, if condition c) is satisfied and $|u_{\mathrm{Eve}}| = |v_{\mathrm{Eve}}|$, AODs 5 to 8 numerically coincide with AODs 1 to 4, respectively, so that the four distinct AODs can be given by AODs 1 to 4.
    
    Finally, if condition c) is violated, there exists no $\tilde{\boldsymbol{\Omega}}$ within the given UPA such that $[\tilde{\boldsymbol{x}}(\tilde{\boldsymbol{\Omega}}), \tilde{\boldsymbol{z}}(\tilde{\boldsymbol{\Omega}})] = [s_1\boldsymbol{z}(\boldsymbol{\Omega}), s_2\boldsymbol{x}(\boldsymbol{\Omega})]$ for any $s_1, s_2 \in \{-1, 1\}$. As a result, AODs 5 to 8 are not valid AOD estimations at Eve, regardless of whether condition a) or b) is satisfied. Note that condition b) does not determine whether AODs 1 to 4 are distinct. In this case, if condition a) is satisfied, AODs 1 to 4 are distinct; otherwise, AODs 1 to 4 collapse to two distinct solutions, as discussed above. Theorem \ref{thm:eight-ambiguity} is thus proved.

\section{Proof of Theorem \ref{thm:Q-level-construction}} \label{app:Q-level-construction}

    By the construction of $\tilde{\boldsymbol{\Lambda}}$ in \eqref{equ:thm-Lambda-construction}, given any $\boldsymbol{\Omega}\in\tilde{\boldsymbol{\Lambda}}$, there exist a strategy $\boldsymbol{\Omega}_0\in\mathcal{R}$, a transformation matrix $\boldsymbol{T}_0\in\mathcal{T}$, and a signed permutation $\boldsymbol{\Pi}_0\in\boldsymbol{\Sigma}$ such that
    \begin{equation} \label{equ:Omega-decomposition}
        \boldsymbol{P}(\boldsymbol{\Omega}) = \boldsymbol{P}(\boldsymbol{\Omega}_0)\boldsymbol{T}_0\boldsymbol{\Pi}_0.
    \end{equation}
    It is worth noting that the region-covering condition (\ref{equ:thm-square-covering}) implies a covering multiplicity of at least $Q$ over the region $\mathbb{R}^2 \setminus (\mathbb{Z}^2 \setminus 2\mathbb{Z}^2)$ due to the periodicity of $2\mathbb{Z}^2$. Therefore, given any $\boldsymbol{w}$ with $0<\|\boldsymbol{w}\|_2 < 1$, define 
    \begin{equation}\label{equ:w0-decomposition}
    \boldsymbol{w}_0 = \boldsymbol{T}_0\boldsymbol{\Pi}_0\boldsymbol{w},
    \end{equation}
    and there always exist at least $Q$ pairs of $\boldsymbol{T}_q \in \mathcal{T}$ and $\boldsymbol{n}_q \in 2\mathbb{Z}^2$ such that
    \begin{equation} \label{equ:elliptic-region-constraints}
        \left\|\boldsymbol{T}_q^{-1}(\boldsymbol{w}_0 + \boldsymbol{n}_q)\right\|_2 < 1, \quad q=1,\ldots,Q.
    \end{equation} 
    For each $q$, because $\boldsymbol{\Omega}_0 \in \mathcal{R}$ satisfies at least one of the constraint sets in (\ref{equ:thm-physical-bound}) and (\ref{equ:thm-physical-bound-2}), there exist a strategy $\tilde{\boldsymbol{\Omega}}_q$ in the given UPA and a signed permutation $\boldsymbol{\Pi}_q \in \boldsymbol{\Sigma}$ such that
    \begin{equation} \label{equ:Omega-tilde-decomposition}
        \tilde{\boldsymbol{P}}(\tilde{\boldsymbol{\Omega}}_q) = \boldsymbol{P}(\boldsymbol{\Omega}_0)\boldsymbol{T}_q\boldsymbol{\Pi}_q.
    \end{equation}
    Since $\boldsymbol{T}_q \in \mathrm{GL}(2,\mathbb{Z})$ and $\boldsymbol{\Omega}_0 \in \boldsymbol{\Lambda}$, it can then be verified that $\tilde{\boldsymbol{\Omega}}_q \in \boldsymbol{\Lambda}$ based on Proposition \ref{prop:uniqueness}. Consequently, $\tilde{\boldsymbol{\Omega}}_q \in \tilde{\boldsymbol{\Lambda}}$ by the construction of $\tilde{\boldsymbol{\Lambda}}$ in \eqref{equ:thm-Lambda-construction}.
    Next, define
    \begin{equation} \label{equ:w-tilde-decomposition}
        \tilde{\boldsymbol{w}}_q = \boldsymbol{\Pi}_q^{-1}\boldsymbol{T}_q^{-1}(\boldsymbol{w}_0 + \boldsymbol{n}_q).
    \end{equation}
    Due to the fact that $\boldsymbol{\Pi}_q \in \boldsymbol{\Sigma}$ is an orthogonal matrix and using (\ref{equ:elliptic-region-constraints}), it follows that $\|\tilde{\boldsymbol{w}}_q\|_2 = \|\boldsymbol{T}_q^{-1}(\boldsymbol{w}_0 + \boldsymbol{n}_q)\|_2 < 1$. Because $\|\boldsymbol{w}\|_2 \neq 0$, it is easy to verify from (\ref{equ:w0-decomposition}) and (\ref{equ:w-tilde-decomposition}) that $\|\tilde{\boldsymbol{w}}_q\|_2 \neq 0$. Hence, $\tilde{\boldsymbol{w}}_q$ corresponds to a valid AOD solution.
    Finally, we have
    \begin{align} \label{equ:received-signal-difference}
        &\tilde{\boldsymbol{P}}(\tilde{\boldsymbol{\Omega}}_q)\tilde{\boldsymbol{w}}_q - \boldsymbol{P}(\boldsymbol{\Omega})\boldsymbol{w}  \stackrel{(a)}{=} \tilde{\boldsymbol{P}}(\tilde{\boldsymbol{\Omega}}_q)\tilde{\boldsymbol{w}}_q - \boldsymbol{P}(\boldsymbol{\Omega}_0)\boldsymbol{w}_0 \nonumber\\
        &\stackrel{(b)}{=} \boldsymbol{P}(\boldsymbol{\Omega}_0)\boldsymbol{T}_q\boldsymbol{\Pi}_q \boldsymbol{\Pi}_q^{-1}\boldsymbol{T}_q^{-1}(\boldsymbol{w}_0 + \boldsymbol{n}_q) - \boldsymbol{P}(\boldsymbol{\Omega}_0)\boldsymbol{w}_0 \nonumber\\
        &= \boldsymbol{P}(\boldsymbol{\Omega}_0)\boldsymbol{n}_q, 
    \end{align}
    where (a) is from (\ref{equ:Omega-decomposition}) and (\ref{equ:w0-decomposition}), and (b) is from (\ref{equ:Omega-tilde-decomposition}) and (\ref{equ:w-tilde-decomposition}).
    Due to the facts that $\boldsymbol{P}(\boldsymbol{\Omega}_0)$ is an integer matrix and $\boldsymbol{n}_q \in 2\mathbb{Z}^2$, from (\ref{equ:received-signal-difference}) it follows that $\tilde{\boldsymbol{P}}(\tilde{\boldsymbol{\Omega}}_q)\tilde{\boldsymbol{w}}_q - \boldsymbol{P}(\boldsymbol{\Omega})\boldsymbol{w} \in 2\mathbb{Z}^K$. 
    According to Proposition 2, this guarantees that the strategy $\tilde{\boldsymbol{\Omega}}_q$ and its corresponding AOD $\tilde{\boldsymbol{w}}_q$ for each $q$ can lead to Eve's received signal. Crucially, due to (\ref{equ:thm-square-non-overlap}), these strategies $\tilde{\boldsymbol{\Omega}}_q$ inherently belong to $Q$ distinct strategy groups, including one basic and $Q-1$ enhanced.
    Theorem \ref{thm:Q-level-construction} is thus proved.

\end{appendices}

\bibliographystyle{IEEEtran}
\bibliography{IEEEabrv,References}

% Biography text here without a photo.

% \begin{IEEEbiography}[{\includegraphics[width=1in,height=1.25in,clip,keepaspectratio]{fig1.png}}]{IEEE Publications Technology Team}
% In this paragraph you can place your educational, professional background and research and other interests.\end{IEEEbiography}

\end{document}